\pdfoutput=1
\documentclass{article}

\usepackage{arxiv}

\usepackage[utf8]{inputenc}
\usepackage[T1]{fontenc}
\usepackage{hyperref}
\usepackage{url,amsfonts,xcolor}
\usepackage{graphicx}
\usepackage{epstopdf}
\usepackage{algorithmic}
\usepackage{amsmath,amssymb,amsthm,enumerate,enumitem,mathtools}
\usepackage{cleveref}

\newtheorem{theorem}{Theorem}[section]
\newtheorem{definition}[theorem]{Definition}
\newtheorem{lemma}[theorem]{Lemma}
\newtheorem{corollary}[theorem]{Corollary}

\newtheorem{claim}[theorem]{Claim}

\newcommand{\E}{\text{\bf E}}
\def \v {\boldsymbol{v}}

\def \e {\boldsymbol{e}}
\def \eps {\epsilon}

\DeclareMathOperator{\I}{\mathcal{I}}
\let\S=\undefined
\DeclareMathOperator{\S}{\mathcal{S}}

\DeclareMathOperator{\F}{\mathcal{F}}
\DeclareMathOperator{\A}{\mathcal{A}}
\DeclareMathOperator{\B}{\mathcal{B}}
\DeclareMathOperator{\C}{\mathcal{C}}

\DeclareMathOperator{\J}{\mathcal{J}}
\DeclareMathOperator{\g}{\mathbf{g}}

\DeclareMathOperator{\SR}{\mathcal{SR}}
\DeclareMathOperator{\ALG}{\mathcal{ALG}}
\DeclareMathOperator*{\argmin}{arg\,min}

\DeclareMathOperator{\OO}{\mathcal{O}}

\ifpdf
  \DeclareGraphicsExtensions{.eps,.pdf,.png,.jpg}
\else
  \DeclareGraphicsExtensions{.eps}
\fi

\title{Improved Bi-point Rounding Algorithms and a Golden Barrier for $k$-Median}

\author{Kishen N Gowda \\
	University of Maryland, College Park, US\\
	\texttt{kishen19@cs.umd.edu} \\
	\And
	Thomas Pensyl \\
	\texttt{tommy.pensyl@gmail.com}\\
	\AND
	Aravind Srinivasan\thanks{Aravind Srinivasan was supported in part by NSF grant CCF-1918749, as well as by research awards from Amazon and Google.} \\
	University of Maryland, College Park, US \\
	\texttt{srin@cs.umd.edu} \\
	\And
	Khoa Trinh \\
	Google \\
	\texttt{khoatrinh@google.com} \\
}

\renewcommand{\headeright}{Improved Bi-point Rounding Algorithms and a Golden Barrier for $k$-Median}
\renewcommand{\undertitle}{}
\renewcommand{\shorttitle}{}
\date{}

\hypersetup{
pdftitle={Improved Bi-point Rounding Algorithms and a Golden Barrier for k-Median},
pdfsubject={Improved Bi-point Rounding Algorithms and a Golden Barrier for k-Median},
pdfauthor={Kishen N Gowda, Thomas Pensyl, Aravind Srinivasan, Khoa Trinh},
pdfkeywords={
k-median, clustering, approximation algorithms}}

\begin{document}
\maketitle

\begin{abstract}
The current best approximation algorithms for $k$-median rely on first obtaining a structured fractional solution known as a \emph{bi-point solution}, and then rounding it to an integer solution. We improve this second step by unifying and refining previous approaches. We describe a hierarchy of increasingly-complex partitioning schemes for the facilities, along with corresponding sets of algorithms and factor-revealing non-linear programs. We prove that the third layer of this hierarchy is a $2.613$-approximation, improving upon the current best ratio of $2.675$, while no layer can be proved better than $2.588$ under the proposed analysis.
  
On the negative side, we give a family of bi-point solutions which cannot be approximated better than the square root of the golden ratio, even if allowed to open $k+o(k)$ facilities. This gives a barrier to current approaches for obtaining an approximation better than $2 \sqrt{\phi} \approx 2.544$. Altogether we reduce the approximation gap of bi-point solutions by two thirds.
\end{abstract}

\section{Introduction}\label{sec:intro}


We study the classical $\mathcal{NP}$-hard $k$-median problem. Given a set of clients $\C$, a set of facilities $\F$, and a distance metric $d$ over $\C \cup \F$, the goal is to \emph{open} a subset $\S$ of $k$ facilities which minimizes the total distance from each client to its closest facility in $\S$. In addition to their natural applications in facility location (e.g., in the opening of health-care facilities),
such problems also naturally model clustering in various ways. 



\subsection{Related Work}

Three main approaches have been applied toward constant factor approximations for $k$-median. The first is LP-rounding, using a half-integral solution as an intermediate step. Charikar et al.\  \cite{charikar2002constant} used this approach to give the first constant-factor approximation of $6\frac{2}{3}$. Charikar and Li \cite{charikar2012dependent} later refined this approach with dependent rounding to achieve a factor of $3.25$. 

The second approach is local search. Arya et al.\ \cite{arya2004local} proved that the natural local search algorithm with $\frac2\eps$ simultaneous swaps is a tight $(3+\eps)$-approximation. Their analysis was later nicely simplified by Gupta and Tangwongsan \cite{gupta2008simpler}. Very recently, Cohen-Addad et al.\ \cite{cohen2022improved} showed that this could actually be improved by using an auxiliary cost function which discounts clients with multiple nearby open facilities. They proved that this algorithm with $\OO\big ((\frac1\eps)$\^{}$(\frac1\eps)$\^{}$(\frac1\eps)\big)$ simultaneous swaps achieves an approximation factor of $2.836+\eps$. 

The third approach, and the one we focus on for the rest of the paper, is to first generate an intermediate fractional form called a \emph{bi-point solution}, and then round it to an integral solution. These two steps are usually analyzed independently, each incurring a multiplicative loss in the approximation factor, which we will refer to as the \emph{bi-point generation factor} and the \emph{bi-point rounding factor} respectively. 
Jain and Vazirani~\cite{jain2001approximation} introduced this approach, achieving a bi-point generation factor of 3 via reduction to the related Uncapacitated Facility Location problem, and a bi-point rounding factor of $2$, resulting in a $6$-approximation for $k$-median. Jain, Mahdian and Saberi~\cite{JMS} later reduced the bi-point generation factor from $3$ to $2$ to get a $4$-approximation; Charikar and Guha also gave an earlier 4--approximation by analyzing both steps holistically \cite{DBLP:conf/focs/CharikarG99}. 

Li and Svensson \cite{LiSven} improved the bi-point rounding factor to $\frac{1+\sqrt3+\eps}{2}\approx 1.366$ under the relaxation of being allowed to open $k+\OO\left(\frac{1}{\epsilon}\right)$ facilities, resulting in a \emph{pseudo-approximation} algorithm. Then, in a remarkable result, they showed how to convert \emph{any} $\alpha$-pseudo-approximation algorithm which opens $k+c$ facilities into a proper $(\alpha+\epsilon)$-approximation algorithm, by running it on a set of $n^{\OO(c/\epsilon)}$ derived instances. This enabled conversion of their pseudo-approximation into a $(1+\sqrt3+\eps)$-approximation with runtime $n^{\OO(1/\eps^2)}$.
Their algorithm was based on constructing \emph{stars} of nearby facilities which could be randomly rounded while still guaranteeing a nearby facility for each client.

Byrka et al.\ \cite{Byrka-et-al} refined this approach by categorizing stars by their size and relative distance, and then considering a slew of algorithms which open facilities asymmetrically across each category. They then gave a factor-revealing non-linear program (NLP) corresponding to taking the best of all algorithms, and solved the NLP with computer assistance to get a bi-point rounding factor of $1.3371 +\eps$. They also re-framed the star-rounding technique in terms of bounding positive correlation, and developed a dependent-rounding technique which reduced the number of facilities required to $k+\OO\left(\log\frac1\eps\right)$, resulting in a runtime of $n^{\OO\left(\frac1\eps\log\frac1\eps\right)}$. They also exhibited a bi-point solution with an integrality gap of $\frac{1+\sqrt2}2\approx1.207$ even when allowing for $k+o(k)$ facilities to be opened, thus giving a new lower bound for the bi-point rounding factor, resilient to Li and Svensson's pseudo-approximation technique.

Lastly, during the writing of this paper, Cohen-Addad et al.~\cite{cohen22improvedLMP} have released results refining the bi-point generation factor to be slightly less than $2$. The remainder of our paper does not reflect this result, though it would likely imply some small improvement to the overall factor.

\subsection{Our Work}
Let us now formally define bi-point solutions and summarize the current-best bi-point generation result.

\begin{definition} [Bi-point solution] \label{def:bipoint} Given a $k$-median instance $\I$, a bi-point solution is a pair $F_1, F_2 \subseteq \F$ such that $|F_1| \leq k \leq |F_2|$, along with real numbers $a,b \geq 0$, $a+b=1$ such that $a|F_1| + b|F_2| = k$. The cost of this bi-point solution is defined as $aD_1 + bD_2$, where $D_1$ and $D_2$ are the total connection costs of $F_1$ and $F_2$, respectively.
\end{definition}

\begin{theorem}[\cite{JMS}]\label{thm:JMS_2BGF}
    There exists an algorithm that, given a $k$-median instance, produces a bi-point solution of cost at most $2\cdot OPT$ in polynomial time, where $OPT$ is the cost of optimal solution for the given instance.
\end{theorem}

In this paper, we will focus exclusively on improving the bi-point rounding factor, where a factor of $\alpha$ immediately implies a $2\alpha$-approximation for $k$-median, per Theorem~\ref{thm:JMS_2BGF}. 

A key ingredient in improving this factor will be the star-rounding algorithm of Li and Svensson (LS). One drawback of their algorithm is that it requires separate algorithms when $b$ is close to $0$ or $1$, and this caveat propagates to any dependent results. To avoid importing this complexity, we provide a streamlined version of LS  that uses a somewhat-sophisticated dependent-rounding procedure to remove any restriction on the value of $b$, by essentially unifying the star-rounding and knapsack-type algorithms from LS.
This more generic analysis requires opening slightly more facilities than LS, but we feel the gain in simplicity is worthwhile and hope others may find it similarly helpful in future work. (Indeed, this technique would have removed the need for many edge-case algorithms and analysis in \cite{Byrka-et-al} as well.)

\begin{theorem}\label{thm:sr}
    There exists a randomized algorithm $\mathcal{SR}$ that takes a bi-point solution and returns a set of $k+\OO\left(\frac1\eps \log\frac1\eps\right)$ facilities, with expected cost at most
    \begin{align*}
        (1+\eps) \cdot \left[(1-b)D_1+b(3-2b)D_2) \right]. 
    \end{align*} 
\end{theorem}

To account for the extra facilities opened, we invoke the pseudo-approximation reduction of LS for a net run-time of $n^{\OO\left(\frac{1}{\eps^2}\log\frac1\eps\right)}$, and an additional additive penalty of $\eps$ incurred in the approximation factor. For convenience, we generally omit the $\eps$ term when stating our approximation factors throughout the paper. (Alternatively, we may consider $\eps$ as being absorbed into the rounding error in our non-exact approximation factors).

Now consider the previously mentioned bi-point rounding algorithms of Jain and Vazirani (JV) and of LS. JV is based on $F_2$-centric stars, while LS is based on $F_1$-centric stars. In each case, the stars function to provide a worst case backup bound for each client.  Byrka et al.\ refined LS by partitioning stars according to a factor $g(i)$ representing  their relative distance to each other. This exposed new possible backup bounds beyond those immediately provided by the stars, thus enabling a larger variety of competing algorithms, and resulting in the improved approximation.

In this paper we apply a similar refinement, but to JV instead of LS. We find JV is more amenable to these techniques for several reasons. First, the resulting backup bounds are more direct and thus cheaper, due to only needing to make at most one ``hop" to another facility. Second, we are able to apply the $g(\cdot)$-based partitioning to \emph{all} stars, instead of just some, which also means the new backup bounds are available to all clients. Third, the resulting algorithms are simpler and do not open more than $k$ facilities.
This last property immediately implies that our new algorithms cannot actually beat the integrality gap of 2. Nevertheless, when run in tandem with LS, we achieve results which appear to completely subsume those in Byrka et al.\ (e.g., including their algorithms in our analysis does not improve the results of this paper, at least experimentally).

Additionally, thanks to the simpler setting of $F_2$-centric stars, we are able to push the technique further, defining and analyzing a hierarchy of increasingly complex partitions based on the factor $g(i)$. For each layer of the hierarchy, we propose a set of candidate algorithms, and a factor-revealing NLP representing the best of all solutions. The NLP complexity increases exponentially at each layer. However with computer-assisted methods, we are able to rigorously prove that the third layer, in tandem with $\mathcal{SR}$, provides a bi-point rounding factor of $1.3064$ 
(improving over the previous best factor of $1.3371$). We also show the existence of a difficult instance for which our NLP cannot prove a bi-point rounding factor smaller than $1.2943$, for any layer of our hierarchy.

\begin{theorem}\label{thm:main_result}
    There exists a randomized algorithm that, given a bi-point solution of cost $aD_1+bD_2$ and $\eps > 0$, opens at most $k+\OO\left(\frac{1}{\epsilon}\log\frac1\epsilon\right)$ facilities and returns a solution of cost at most $(1.3064+\epsilon)\cdot(aD_1+bD_2)$.
\end{theorem}
\begin{theorem}\label{thm:lower_bound}
    There exists a bi-point solution for which our NLP cannot prove a bi-point rounding factor better than $1.2943$, for any layer of our hierarchy.
\end{theorem}

Theorem~\ref{thm:JMS_2BGF}, Theorem~\ref{thm:main_result} and the pseudo-approximation reduction of Li and Svensson~\cite[Theorem~4]{LiSven} together imply our improved approximation factor.

\begin{corollary}
    There exists a randomized algorithm that, given a $k$-median instance and $\eps > 0$, runs in time $n^{\OO\left(\frac{1}{\eps^2}\log\frac1\eps\right)}$ and returns a solution of cost at most $(2.613+\eps)\cdot OPT$, where $OPT$ is the cost of the optimal solution of the given instance.
\end{corollary}

Lastly, we provide a bi-point solution with integrality gap $\sqrt\phi\approx1.272$ where $\phi$ is the golden ratio, even when allowing for $k+o(k)$ facilities to be opened. This improves on the previous gap of $1.207$, giving an improved lower bound for the bi-point approximation factor. Study of this instance inspired the algorithmic improvements in this paper; we hope it can shed further insight into the true approximability of bi-point solutions.

\begin{theorem}\label{thm:integrality_gap}
For every $\eps>0$ and $C(k) =o(k)$, there exists a family of bi-point solutions with integrality gap $\sqrt{\phi}-\eps$, even if we allow solutions that open $k+C(k)$ facilities.
\end{theorem}

The rest of the paper is organized as follows. In \Cref{sec:algorithms}, we present our bi-point rounding algorithm. \Cref{sec:star_round} discusses the aforementioned variant of Li and Svensson's star-rounding algorithm $\mathcal{SR}$, \Cref{sec:main} describes our main hierarchy of partitioning schemes for facilities and the corresponding set of rounding algorithms, and \Cref{sec:nlp} presents the NLP for obtaining the bi-point rounding factor and the results achieved for the second and third layers of our hierarchy. In \Cref{sec:lower_bounds}, we give a lower bound for our framework. In \Cref{sec:integrality_gap}, we exhibit the pseudo-approximation-resilient family of integrality-gap instances. We conclude with a discussion in \Cref{sec:discussion}.

\section{Bi-point Rounding Algorithm}\label{sec:algorithms}

In this section, we define a set of bi-point rounding algorithms, and bound the total cost and facilities opened for each algorithm. Our set of algorithms consists of a family of rounding algorithms obtained via a hierarchy of partition schemes for the facilities, along with the star-rounding algorithm $\mathcal{SR}$. Our top-level algorithm runs all of these algorithms and returns the best solution obtained.

\subsection{Preliminaries}

For any client $j \in \C$, we use $c(j)$ to denote the connection cost of $j$ (to its closest open facility in $\S$). For any set $J$ of clients, let $c(J) := \sum_{v \in J} c(v)$.  We let $i_1(j)$ and $i_2(j)$ denote the closest facility to $j$ in $F_1$ and $F_2$, respectively. When the context is clear, we shall drop the parameter $j$ in the notation.

By an abuse of notation, for any facility $i \in \F$, we let $i$ (and $\bar{i}$) denote the event that facility $i$ is opened (and closed) in our solution (respectively). For ease of notation, we also drop the $\wedge$ operator when joining these events. For example, $\Pr[i_1 \bar{i_2}]$ is the probability for the event that facility $i_1$ is opened and facility $i_2$ is closed. For any set $X \subseteq \F$, we let $\sigma_X(i)$ denote the closest facility to $i$ in $X$ (i.e., $\sigma_X(i) := \argmin_{i' \in X} d(i, i')$).

Lastly, we use $[m]$ to denote the set of natural numbers $\{1,2,\cdots, m\}$, and $[m_1,m_2]$ to denote the set $\{m_1,m_1+1,\cdots,m_2\}$. However, when we refer explicitly to the set $[0,1]$, it means the standard set of reals between $0$ and $1$ (inclusive).

\subsection{Star-Rounding Algorithm}\label{sec:star_round}
In this section, we define the star-rounding algorithm $\mathcal{SR}$ and prove Theorem~\ref{thm:sr}, which will provide the same cost guarantee as Li and Svensson's algorithm, but with no restrictions on bi-point parameter $b$ (see \Cref{def:bipoint}). The proof follows a similar structure to that of previous star-rounding algorithms\cite{LiSven,Byrka-et-al}. We will first need the following dependent-rounding procedure from \cite{SRDR2017} \footnote{This result appears exclusively in v1 of the referenced arXiv paper, but continues to be publicly available.}.
\begin{theorem}[{\cite[Theorem~2.1]{SRDR2017}}] \textbf{Symmetric Randomized Dependent Rounding.}  
\label{srdr}
Given vectors $x \in [0,1]^n$ and $a = (a_1, \ldots, a_n) \in \mathbb{R}^n$, and $t \in \mathbb{N}$, there exists a randomized algorithm (SRDR) which runs in expected $\OO(n^2)$ time and returns a vector $X \in [0,1]^n$ with at most $t$ fractional values. Both the weighted sum and all the marginal probabilities are preserved: $\sum_i a_i X_i = \sum_i a_i x_i$ with probability one, and $\E[X_i] = x_i$ for all $i \in [n]$. Let $S, T$ be disjoint subsets of $[n]$. Then we have the upper correlation bound:
\begin{equation}\label{eq:SRDR}
\E \left[\prod_{i\in S}X_i\prod_{j \in T}(1-X_j) \right] 
\le \left(\prod_{i\in S}x_i\prod_{j \in T}(1-x_j) \right)^{1 - 1/(t+1)}.
\end{equation}
\end{theorem}
This can generally be converted to a standard multiplicative $(1+\eps)$ error bound by setting the parameter $t$ to be $\OO\left(\frac1\eps \log \frac{1}{\prod_{i\in S}x_i\prod_{j \in T}(1-x_j)}\right)$, but in the special case below, we can actually set $t$ independently of the $x_i$. This will allow us to avoid restricting the domain of the algorithm.
\begin{corollary}\label{corr:SRDRx} \textbf{Pairwise positive correlation under uniform marginals.}
Suppose SRDR is run with $t\ge\frac{\log(1+1/\eps)}{\log(1+\eps)}$. For any distinct $i,j$ such that $x_i=x_j=b$ for any value of $b$, we have
\begin{equation*}\label{eq:thm-upper}
\E [X_i(1-X_j) ] \le (1+\eps)b(1-b).
\end{equation*}
\end{corollary}
\begin{proof}
Let $S=\{i\}$, $T=\{j\}$. Then (\ref{eq:SRDR}) simplifies to:
\begin{equation}\label{eq:corrcase1} \E [X_i(1-X_j) ] 
\le (x_i(1-x_j))^{1-1/(t+1)} 
=\left(\frac1{b(1-b)}\right)^{1/(t+1)} b(1-b).
\end{equation}
If $\min\{b,1-b\}\ge\frac\eps{1+\eps}$, then $\frac{1}{b(1-b)}\le \frac{(1+\eps)^2}{\eps}$  and $\frac1{t+1}\le \frac{\log(1+\eps)}{\log((1+\eps)^2/\eps)}$, so (\ref{eq:corrcase1}) is at most $(1+\eps)b(1-b)$. Otherwise $\min\{b,1-b\}\le\frac\eps{1+\eps}$, in which case we may directly bound by the marginals:
\begin{align}
    \E [X_i(1-X_j) ]
    &\le \min\{\E[X_i],\E[1-X_j]\}\label{eq:relax-marginals}\\
    &= \min\{b,1-b\}
    = \frac{b(1-b)}{1-\min\{b,1-b\}}
    \le (1+\eps)b(1-b).\nonumber
\end{align}
\end{proof}
We now define the star-rounding algorithm $\mathcal{SR}$. Form graph $G$ by drawing an edge from each facility in $F_2$ to its closest facility in $F_1$, resulting in a forest of $F_1$-centric stars. For each $i\in F_1$, let $L_i\subseteq F_2$ be the leaves of the star rooted at $i$.  Define vectors
$a,x$ with $a_i:=|L_i|-1$ and $x_i:=b$, and set $t:=\big\lceil\frac{\log(1+1/\eps)}{\log(1+\eps)}\big\rceil$ for some $\eps>0$. Now run $\text{SRDR}(a,x,t)$ to obtain output vector $X$. Finally, for each $i\in F_1$, open $\lceil X_i |L_i|\rceil$ facilities uniformly at random from $L_i$ and open $i$ itself with probability $\lceil 1-X_i\rceil$.

\begin{lemma}\label{lem:star-round_fac}
$\mathcal{SR}$ opens at most $k+\OO\left(\frac1\eps \log\frac1\eps\right)$ facilities, with probability one. 
\end{lemma}
\begin{proof}
By Theorem~\ref{srdr}, there are at most $t$ fractionally-valued $X_i$ (the rest being $0$ or $1$). Also, $\sum_i (|L_i|-1) X_i=\sum_i a_i X_i = \sum_i a_i x_i = \sum_i (|L_i|-1)b$ and $\{L_i\}_{i\in F_1}$ partitions $F_2$.
Thus, the count of facilities opened is 
\begin{align*}
\sum_{i\in F_1}\big(\lceil 1-X_i\rceil+\lceil X_i |L_i|\rceil\big)
&\le\sum_{i\in F_1}( 1-X_i + X_i |L_i|) + 2t ~~\mbox{(with probability one)}
\\&=\sum_{i\in F_1}( 1-b +  b |L_i|)) + 2t ~~\mbox{(with probability one)}
\\&=(1-b) |F_1| + b |F_2| + 2t
\\&=k + 2t = k+\OO\left(\frac1\eps \log\frac1\eps\right).
\end{align*}
\end{proof}

\begin{lemma}
For any two facilities $i_1\in F_1$ and $i_2\in F_2$, we have
\begin{align}
\Pr[\bar i_1]\le& b, \\
\Pr[\bar i_2]\le& 1-b, \\
\Pr[\bar i_1\bar i_2]\le& (1+\eps)b(1-b).
\end{align}
\end{lemma}
\begin{proof}
We shall use this simple fact: for any event $\mathcal E$ and random variable $Y$, $\Pr[\mathcal E|Y=y]\le f(y)$ implies that $\Pr[\mathcal E]\le\E[f(Y)]$. Let $i_3$ be the root of the star containing $i_2$. Then we have that 
\begin{itemize}
\item $\Pr[\bar i_1|X_{i_1}=y]=1-\lceil 1-y\rceil \le y\implies\Pr[\bar i_1]\le\E[X_{i_1}]=b;$

\item $\Pr[\bar i_2|X_{i_3}=z]=1-\frac{\lceil z |L_i|\rceil}{|L_i|}\le 1-z \implies \Pr[\bar i_2]\le\E[1-X_{i_3}]=1-b;$

\item If $i_1$ and $i_2$ are in the same star, then  ``$i_1$ is closed '' $\implies X_{i_1}=1 \implies$ ``$i_2$ is opened'', so $\Pr[\bar i_1\bar i_2]=0$. Else the two facilities are chosen by independent processes (for fixed $X$), and we can apply \Cref{corr:SRDRx} to show
\[ \Pr[\bar i_1\bar i_2|X_{i_1}=y\land X_{i_3}=z]
=\Pr[\bar i_1|X_{i_1}=y] \cdot \Pr[\bar i_2| X_{i_3}=z]
\le y(1-z),
\] 
implying that
\begin{align*}
\Pr[\bar i_1\bar i_2] \le \E[X_{i_1}(1-X_{i_3})]\le (1+\eps)b(1-b).
\end{align*}
\end{itemize}
\end{proof}

\begin{lemma}\label{lem:star_round_cost}
For any client $j$, let $i_1$ and $i_2$ be the closest facilities in $F_1$ and $F_2$, at distances $d_1$ and $d_2$ from $j$, respectively. 
$\mathcal{SR}$ gives $$\E[c(j)]\le (1+\eps) ((1-b)d_1+b(3-2b)d_2).$$
\end{lemma}
\begin{proof}
Let $i_3\in F_1$ be the root of $i_2$ in $G$. At least one of $i_2$ or $i_3$ will be opened.
By the triangle inequality, $d(j,i_3)\le d(j,i_2)+d(i_2,i_3)\le d(j,i_2)+d(i_2,i_1)\le 2d_2+d_1$. If $d_2<d_1$, we bound the cost by connecting to the first open facility in order of precedence
 $(i_2,i_1,i_3)$:
\begin{align*}
\E[c(j)]&\le \Pr[i_2]d_2+\Pr[\bar i_2i_1]d_1+\Pr[\bar i_2\bar i_1](2d_2+d_1)\\
&= d_2+\Pr[\bar i_2](d_1-d_2)+2\Pr[\bar i_2\bar i_1]d_2\\
&\le d_2+(1-b)(d_1-d_2)+2(1+\eps)b(1-b)d_2\\
&= (1-b)d_1+b d_2+(1+\eps)b(2-2b)d_2.
\end{align*}
Similarly, if $d_1<d_2$, we connect in order $(i_1,i_2,i_3)$:
\begin{align*}
\E[c(j)]&\le \Pr[i_1]d_1+\Pr[\bar i_1 i_2]d_2+\Pr[\bar i_1\bar i_2](2d_2+d_1)\\
&= d_1+\Pr[\bar i_1](d_2-d_1)+\Pr[\bar i_2\bar i_1](d_1+d_2)\\
&\le d_1+b(d_2-d_1)+(1+\eps)b(1-b)(2d_2)\\
&= (1-b)d_1+b d_2+(1+\eps)b(2-2b)d_2.
\end{align*}
\end{proof}
Summing the expected cost of all clients yields Theorem~\ref{thm:sr}. Finally, we remark that $\mathcal{SR}$ generalizes Li and Svensson's knapsack algorithm in the following senses: Firstly, if we run $\mathcal{SR}$ with $t=1$ we essentially recover a randomized version of the knapsack algorithm. Secondly, the knapsack cost analysis works by giving up on ever connecting to $i_1$. In other words, we relax $\Pr[\bar i_2\bar i_1]\le\Pr[\bar i_2]$ (doing so in the above equations indeed recovers the knapsack cost bound), which is equivalent to the relaxation done in \Cref{eq:relax-marginals}.

\subsection{Main Family of Algorithms}\label{sec:main}
In this section, we describe our main hierarchy of increasingly complex partitioning schemes for the facilities. Given a bi-point solution $aF_1 + bF_2$, let us rename the set $F_1$ to $A$ for convenience. Now, we associate each facility $i \in A$ to its nearest facility $\sigma_B(i) \in F_2$ (breaking ties arbitrarily). This gives us a set of ``primary stars’’, where the centers are facilities in $F_2$ and the leaves are the facilities in $A$. Let $B$ denote the set of centers of primary stars with at least one leaf, i.e., $B$ is the set of facilities in $F_2$ that have at least one facility from $A$ associated with it. We pad the set $B$, arbitrarily, with the remaining facilities from $F_2$ until $|B|=|A|$.

Now, let $C$ denote the set $F_2 \setminus B$. We also associate each facility $i \in A$ to its nearest facility $\sigma_C(i) \in C$ (breaking ties arbitrarily). This gives us a set of ``secondary stars'' with centers in $C$ and leaves in $A$. (See \Cref{fig:star_construction}). Also, for a set $S$, let $\sigma_B(S) := \{\sigma_B(x)| x\in S \}$. $\sigma_C(S)$ is defined similarly.

\begin{figure}[htpb]
\centering 
\includegraphics[width=0.8\textwidth]{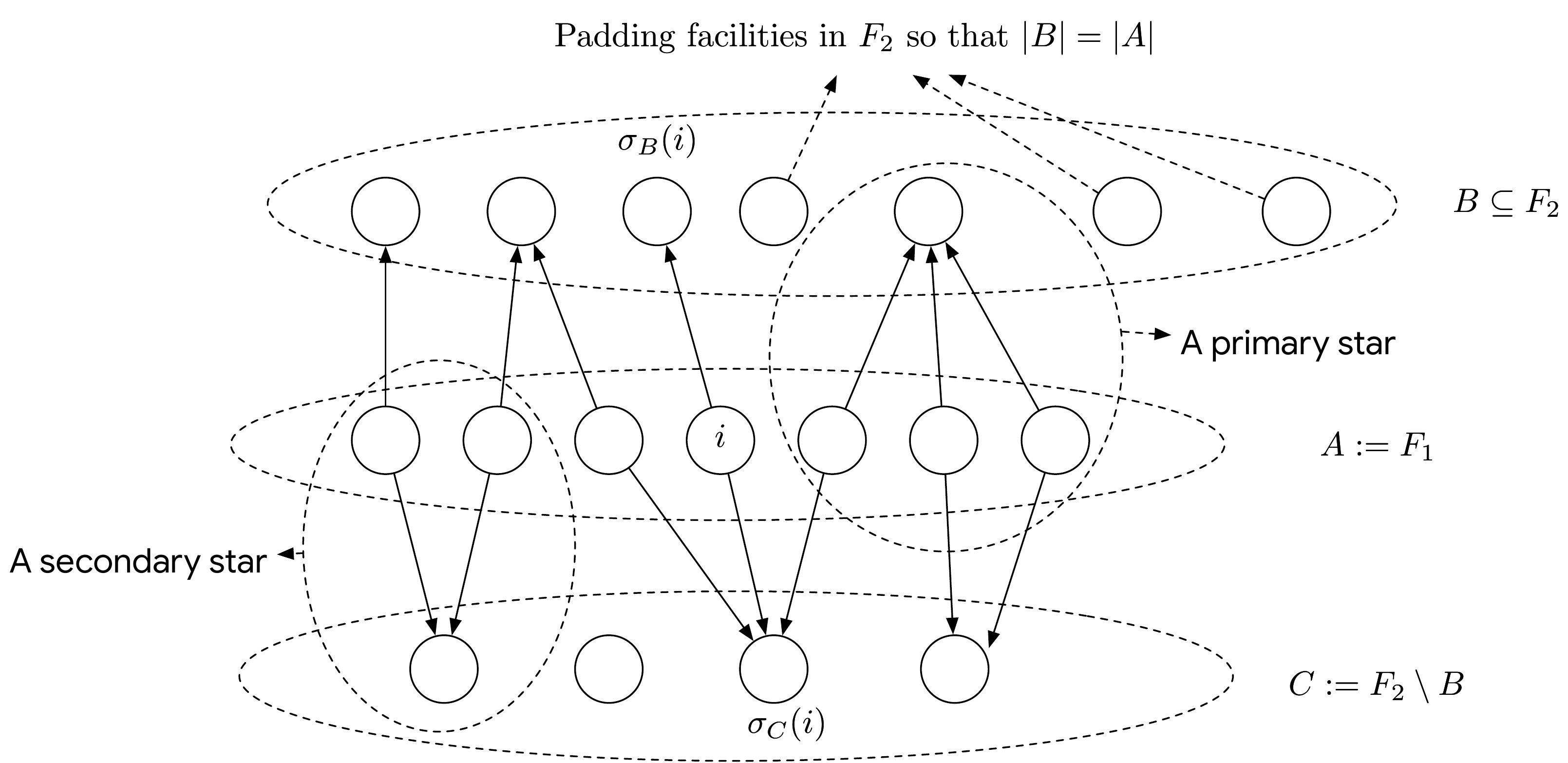}
\caption{Star construction.}\label{fig:star_construction}
\end{figure}

Note that the primary and secondary stars constructed are $F_2$-centric stars, in contrast to the $F_1$-centric stars of Li and Svensson~\cite{LiSven} and Byrka et al.~\cite{Byrka-et-al}. Jain and Vazirani in their work~\cite{jain2001approximation} essentially formulated $F_2$-centric primary stars, like in our construction, and ensured that for every client $j$ at least one of $i_1$ or $\sigma_B(i_1)$ is opened, yielding a $2$-approximation. In our construction, with help of both primary and secondary stars, we ensure that at least one of $i_1$ or $\sigma_B(i_1)$ or $\sigma_C(i_1)$ is opened. Thus, with two backups, we have more freedom when designing our rounding algorithm. Furthermore, the two backups (i.e., $\sigma_B(i_1)$ and $\sigma_C(i_1)$) are relatively closer to the client $j$ and require fewer triangle inequality ``hops'' compared to the backups used by Byrka et al.~\cite{Byrka-et-al}. 

However, $\sigma_C(i_1)$ may not always be a useful backup. We quantify this by introducing the following parameter. For every facility $i \in A$, let
\begin{align*}
    g(i) := \frac{d(i,\sigma_B(i))}{d(i,\sigma_C(i))}.
\end{align*}
Observe that $0 \le g(i) \le 1, \forall i \in A$. In general, if client $j$ is connected to $\sigma_B(i_1)$, by triangle inequality, the connection cost for $j$ would be,
\begin{align*}
    c(j) = d(j,\sigma_B(i_1)) \le d(j,i_1)+d(i_1,\sigma_B(i_1)) \le d(j,i_1)+d(i_1,i_2) \le d_1+(d_1+d_2).
\end{align*}
However, if $i_2 \in C$, then $d(i_1,\sigma_C(i_1)) \le d(i_1,i_2)$. Therefore, 
\begin{align*}
    c(j) &= d(j,\sigma_B(i_1)) \le d(j,i_1)+d(i_1,\sigma_B(i_1)) = d(j,i_1)+g(i_1)d(i_1,\sigma_C(i_1))\\ 
    &\le d_1+g(i_1)(d_1+d_2).    
\end{align*}
Thus, if $g(i)$ is small, we can utilize this better bound on the cost. However, if $i_2 \in B$ and $j$ is connected to $\sigma_C(i_1)$, then the connection cost for $j$ would be, 
\begin{align*}
    c(j) &= d(j,\sigma_C(i_1)) \le d(j,i_1)+d(i_1,\sigma_C(i_1)) = d(j,i_1)+\frac{1}{g(i_1)}d(i_1,\sigma_B(i_1))\\
    &\le d_1+\frac{1}{g(i_1)}(d_1+d_2).
\end{align*}
Hence, it is favorable to consider this backup only when $\frac{1}{g(i)}$ is small.

We now partition the set $A$ as follows: choose $m-1$ distinct values $g_1,g_2,\cdots,g_{m-1}$ such that $g_0:=0 < g_1 < g_2 < \cdots < g_{m-1} <g_m:=1$. Partition the set $A$ as $A = A_1\uplus A_2 \uplus \cdots \uplus A_{m}$, such that $\forall x \in A_t$, $g_{t-1}\le g(x) \le g_t$, $\forall t\in[m]$ (ties broken arbitrarily). We create a corresponding partition of the sets $B$ and $C$ as follows: partition the set $B$ as $B = B_1\uplus B_2 \uplus \cdots \uplus B_{m}$, where $B_1 = \sigma_B(A_1)$ and $B_t = \sigma_B(A_t)\setminus \cup_{s=1}^{t-1}B_s$, for $t \in [2,m]$. The sets $B_1,\cdots,B_{m}$ are padded with the remaining facilities in $B$ such that $|A_t|=|B_t|, \forall t\in[m]$ (this can be done since $|A|=|B|$). Also, partition the set $C$ as $C = C_1\uplus C_2 \uplus \cdots \uplus C_{m}$, where $C_{m} = \sigma_C(A_{m})$, $C_t = \sigma_C(A_t)\setminus \cup_{s=t+1}^{m}C_s$ for $t \in [2,m-1]$, and $C_1 = C\setminus \cup_{s=2}^{m}C_s$. Note that the set $C_{m}$ is defined first and we pad the set $C_t$ until $|A_t|=|C_t|$ before defining the set $C_{t-1}$, for $t \in [2,m]$. Thus, it is possible for some of the $C_t$'s to be empty, and there might exist at most one non-empty set $C_t$ such that $|A_t|\ne |C_t|$ (see \Cref{fig:star_partition}). 

We also define $\gamma_{A_t} = \frac{|A_t|}{|C|}$ and $\gamma_{C_t} = \frac{|C_t|}{|C|}, \forall t\in [m]$. Based on the above construction, observe that $\gamma_{C_{m}} = \min\{1,\gamma_{A_{m}}\}$, $\gamma_{C_t} = \min\{\gamma_{A_t},1-\sum_{s=t+1}^{m}\gamma_{C_s}\}$ for $t \in [2,m-1]$, and $\gamma_{C_1} = 1-\sum_{s=2}^{m}\gamma_{C_s}$.

\begin{figure}[htpb]
\centering 
\includegraphics[width=0.8\textwidth]{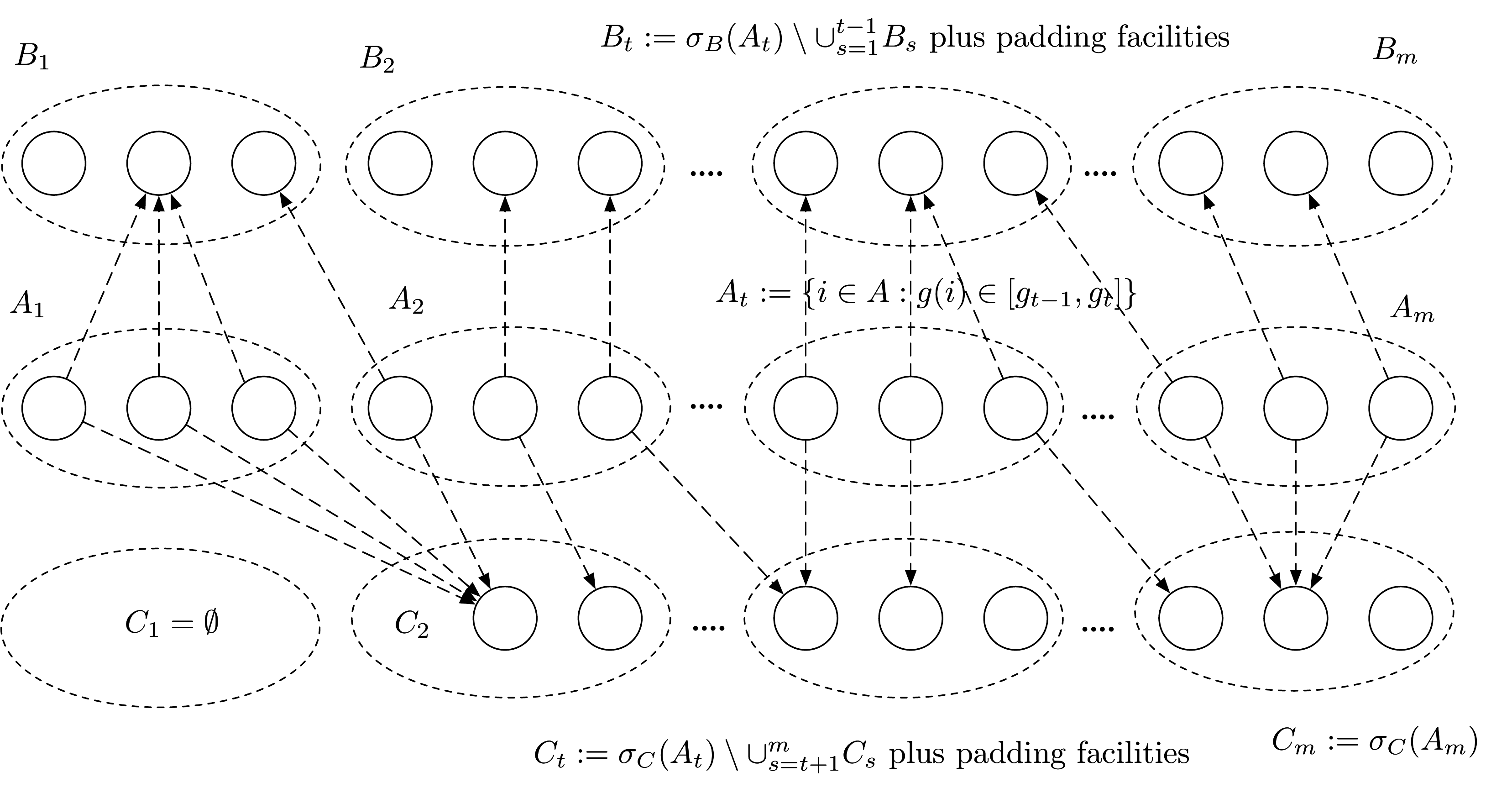}
\caption{Star partition.}\label{fig:star_partition}
\end{figure}

Depending on the value of $m$, we get the different levels of our hierarchy. As discussed above, ideally our cost function would utilize the $g(i)$ and $\frac{1}{g(i)}$ bounds for a facility $i \in A$. Based on the above construction, let $g_{t-1} \le g(i) \le g_t$. Hence, we will instead have to use $g_{t}$ and $\frac{1}{g_{t-1}}$ in place of $g(i)$ when bounding the cost. Thus, considering higher levels of the hierarchy gives us a tighter bound for the associated cost function.

\subsubsection{Algorithm Definition}\label{sec:algo_def}
Given the parameters $\g:=\{g_0,g_1,\cdots,g_{m-1},g_{m}\}$, let $P := \{A_1$, $\cdots$, $A_{m}$, $B_1$, $\cdots$, $B_{m}$, $C_1$, $\cdots$, $C_{m}\}$ denote the partition obtained from the above construction. We now describe the family of rounding algorithms corresponding to this partition. Let $\mathcal{A}(p_{A_1}$, $\cdots$, $p_{A_{m}}$, $p_{B_1}$, $\cdots$, $p_{B_{m}}$, $p_{C_1}$, $\cdots$, $p_{C_{m}})$ be an algorithm that, given the input parameters, uniformly at random samples $\lceil p_W|W| \rceil$ facilities from each set $W \in P$ and returns as output the union of these samples. All algorithms in the family of rounding algorithms that we propose will be of the form of algorithm $\mathcal{A}$ described above, i.e., we do not need any \emph{star-rounding} sub-routine as in \cite{LiSven,Byrka-et-al}. We now define the following notion of ``valid'' algorithms.

\begin{definition}[Valid Algorithm]\label{def:valid_algos}
The algorithm $\mathcal{A}(p_{A_1}$, $\cdots$, $p_{A_{m}}$, $p_{B_1}$, $\cdots$, $p_{B_{m}}$, $p_{C_1}$, $\cdots$, $p_{C_{m}})$ is valid if the following conditions are satisfied:
\begin{enumerate}
    \item $0 \le p_W \le 1$ for each set $W$ in the partition $P$.
    \item Total mass is preserved, i.e., 
    \begin{align}\label{eq:mass_eqn1}
        \sum_{t=1}^{m} p_{A_t}|A_t| + p_{B_t}|B_t| +  p_{C_t}|C_t| = a|F_1|+b|F_2| = k.
    \end{align}
    \item For each $t\in[m]$, either $p_{A_t}=1$, or $p_{B_1}=p_{B_2}=...=p_{B_t}=1$, or $p_{C_t}=p_{C_{t+1}}=...=p_{C_{m}}=1$. This guarantees that for each facility $i\in A$, at least one of $\{i,\sigma_B(i),\sigma_C(i)\}$ is opened. 
    \item At least one of $p_{A_1}$ or $p_{B_1}$ is set to $1$. This guarantees that for each facility $i\in A_1$, at least one of $\{i,\sigma_B(i)\}$ is opened.
\end{enumerate}
\end{definition}
While the set of possible valid algorithms is infinite, our cost function tends to be minimized by extreme point of the parameter space.
Thus, we restrict our attention to the following discrete set. 
\begin{definition}[$\mathcal{ALG}_m$]
For fixed $m$, $\mathcal{ALG}_m$ denotes the set of all algorithms of the form $\mathcal{A}(p_{A_1}$, $\cdots$, $p_{A_{m}}$, $p_{B_1}$, $\cdots$, $p_{B_{m}}$, $p_{C_1}$, $\cdots$, $p_{C_{m}})$ such that,
\begin{itemize}
    \item $\mathcal{A}$ is valid,
    \item At most one of the input parameters of $\mathcal{A}$ is a fractional value (others being either $0$ or $1$).
\end{itemize}
\end{definition}
Note that $\mathcal{ALG}_m$ may be enumerated by recognizing it is a subset of all ways to assign $0$ or $1$ to all but one parameter, and choosing the remaining parameter such that \Cref{eq:mass_eqn1} is satisfied. 
Specifically, this implies $|\mathcal{ALG}_m|\le m2^{3m-1}$. Since $\mathcal{A}(\cdot)$ itself may be implemented in time linear in the number of facilities, then for any fixed $m$, all algorithms in $\mathcal{ALG}_m$ may be run in linear time.

\subsubsection{Bounding the number of facilities}
Our bi-point rounding algorithm returns one (best) of the solutions obtained by algorithms in $\mathcal{ALG}_m$ and $\mathcal{SR}$. From \Cref{lem:star-round_fac}, $\mathcal{SR}$ opens at most $k + \OO\left(\frac{1}{\epsilon}\log\frac{1}{\epsilon}\right)$ facilities. We now show that each algorithm in $\mathcal{ALG}_m$ opens at most $k$ facilities. First, we have the following claim,

\begin{claim}\label{claim:num_fac}
    For algorithms in $\mathcal{ALG}_m$ and for all $W \in P$, $\lceil p_W |W| \rceil = p_W |W|$.
\end{claim}
\begin{proof}
    Recall that, by definition, each algorithm $\mathcal{A}(p_{A_1}$, $\cdots$, $p_{A_{m}}$, $p_{B_1}$, $\cdots$, $p_{B_{m}}$, $p_{C_1}$, $\cdots$, $p_{C_{m}})$ in $\mathcal{ALG}_m$ has at most one input parameter that is fractional. Thus, if $p_W \in \{0,1\}$, then $\lceil p_W|W|\rceil = p_W|W|$. Let $p_V \in [0,1]$ for some $V \in P$. Then, $p_W \in \{0,1\}$ for all $W\in P\setminus \{V\}$. Then, by \Cref{eq:mass_eqn1}, we have
    \begin{align}\label{eq:claim_fac}
        p_V|V| = k - \sum_{W\in P\setminus \{V\}} p_W |W|.
    \end{align}
    Since the right hand side of the above is a non-negative integer, $p_V|V|$ is also a non-negative integer. Therefore, $\lceil p_V|V|\rceil = p_V|V|$.
\end{proof}

\begin{lemma}\label{lem:main_fac}
Each algorithm in $\mathcal{ALG}_m$ opens at most $k$ facilities.
\end{lemma}
\begin{proof}
The number of facilities opened by each algorithm $\mathcal{A}(p_{A_1}$, $\cdots$, $p_{A_{m}}$, $p_{B_1}$, $\cdots$, $p_{B_{m}}$, $p_{C_1}$, $\cdots$, $p_{C_{m}})$ in $\mathcal{ALG}_m$ is,
\begin{subequations}
\begin{align*}
    \sum_{t=1}^m \lceil p_{A_t}|A_t| \rceil + \lceil p_{B_t}|B_t|\rceil + \lceil p_{C_t}|C_t|\rceil = \sum_{t=1}^m p_{A_t}|A_t| + p_{B_t}|B_t| + p_{C_t}|C_t| = k,
\end{align*}
\end{subequations}
where the first equality follows by \Cref{claim:num_fac}, and the second equality follows by \Cref{eq:mass_eqn1}.
\end{proof}

Thus, \Cref{lem:main_fac} and \Cref{lem:star-round_fac} proves the number of facilities bound in Theorem~\ref{thm:main_result}.

\subsubsection{Cost Analysis}\label{sec:cost}
In this section, we derive bounds on the expected connection cost of each client. For a client $j$, we have defined $i_1(j)$ and $i_2(j)$. Let $i_3(j) = \sigma_B(i_1(j))$ and $i_4(j) = \sigma_C(i_1(j))$ (when clear from context, we omit the variable $j$). Assuming $i_1 \in A_t$, for some $t \in [m]$, we define the following cost functions:
\begin{enumerate}
    \item $\C_1(j) := \Pr[i_2]d_2 + \Pr[\bar{i_2}]d_1 + \Pr[\bar{i_1}]\Pr[\bar{i_2}](d_1+d_2)$
    \item $\C_\frac{1}{g}(j) := \Pr[i_2]d_2 + \Pr[\bar{i_2}]d_1 + \frac{1}{g_{t-1}}\Pr[\bar{i_1}]\Pr[\bar{i_2}](d_1+d_2)$
    \item $\C_{1,\frac{1}{g}}(j) := \Pr[i_2]d_2 + \Pr[\bar{i_2}]d_1 + \Pr[\bar{i_1}]\Pr[\bar{i_2}](1+\Pr[\bar{i_3}]\left(\frac{1}{g_{t-1}}-1\right))(d_1+d_2)$.
    \item $\C_{g,1}(j) := \Pr[i_2]d_2 + \Pr[\bar{i_2}]d_1 + \Pr[\bar{i_1}]\Pr[\bar{i_2}](g_{t}+\Pr[\bar{i_3}](1-g_{t}))(d_1+d_2)$
\end{enumerate}

We now have the following lemma that holds for all algorithms in $\mathcal{ALG}_m$.

\begin{lemma}\label{lem:client_cost}
The expected connection cost of a client $j$ is bounded above by,
\begin{enumerate}
    \item $\C_1(j)$ if $i_1 \in A_1$ and $i_2 \in B_t$, for $t \in [m]$,
    \item $\C_{\frac{1}{g}}(j)$ if $i_1 \in A_s$, for $s \in [2,m]$, and $i_2 \in B_t$, for $t \in [s]$,
    \item $\C_{1,\frac{1}{g}}(j)$ if $i_1 \in A_s$, for $s \in [2,m-1]$, and $i_2 \in B_t$, for $t \in [s+1,m]$,
    \item $\C_{g,1}(j)$ if $i_1 \in A_s$, for $s \in [m]$, and $i_2 \in C_t$, for $t \in [m]$.
\end{enumerate}
\end{lemma}
\begin{proof}
    Recall that, by triangle inequality, $d(j,i_3) \le d(j,i_1)+d(i_1,i_3) \le d(j,i_1)+d(i_1,i_2) \le 2d_1+d_2$. We will make heavy use of the fact that facilities in separate partitions are chosen independently.

    In case (1), $i_1\in A_1 \implies i_3 \in B_1$ (by construction). Since $\mathcal{ALG}_m$ contains only algorithms that open at least one of the sets $A_1$ or $B_1$ completely, this guarantees that at least one of $i_1$ or $i_3$ will be opened. 

\begin{figure}[htpb]
\centering 
\includegraphics[width=0.5\textwidth]{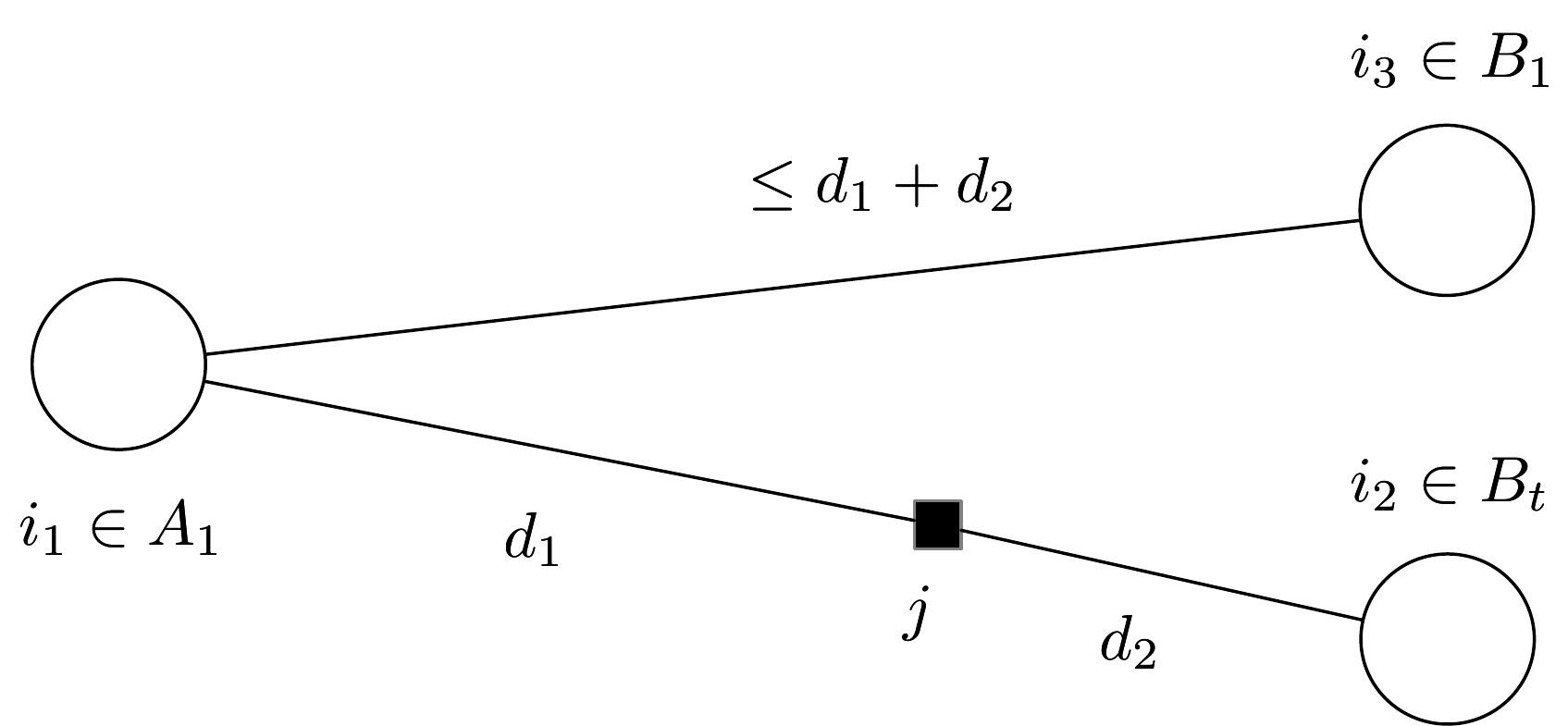}
\caption{Case 1.}
\end{figure}
    
    Therefore, by connecting to the first available facility in the order of precedence $(i_2,i_1,i_3)$,
    \begin{align*}
        \E[c(j)] &\le \Pr[i_2]d_2 + \Pr[\bar{i_2}i_1]d_1 + \Pr[\bar{i_2}\bar{i_1}](2d_1 + d_2)\\
         &= \Pr[i_2]d_2 + \Pr[\bar{i_2}]\Pr[i_1]d_1 + \Pr[\bar{i_2}]\Pr[\bar{i_1}](2d_1 + d_2)\\
         &= \Pr[i_2]d_2 + \Pr[\bar{i_2}]d_1 + \Pr[\bar{i_2}]\Pr[\bar{i_1}](d_1 + d_2) = \C_1(j).
    \end{align*}

    Now, consider case (3). Since $i_1 \in A_s$, $g_{s-1} \le g(i_1)$ $\implies$ $d(i_1,i_4) \le \frac{1}{g_{s-1}}d(i_1,i_3) \le \frac{1}{g_{s-1}}(d_1+d_2)$. Hence, by triangle inequality, $d(j,i_4) \le d(j,i_1)+d(i_1,i_4) \le d_1 + \frac{1}{g_{s-1}}(d_1+d_2)$. 

\begin{figure}[htpb]
\centering 
\includegraphics[width=0.75\textwidth]{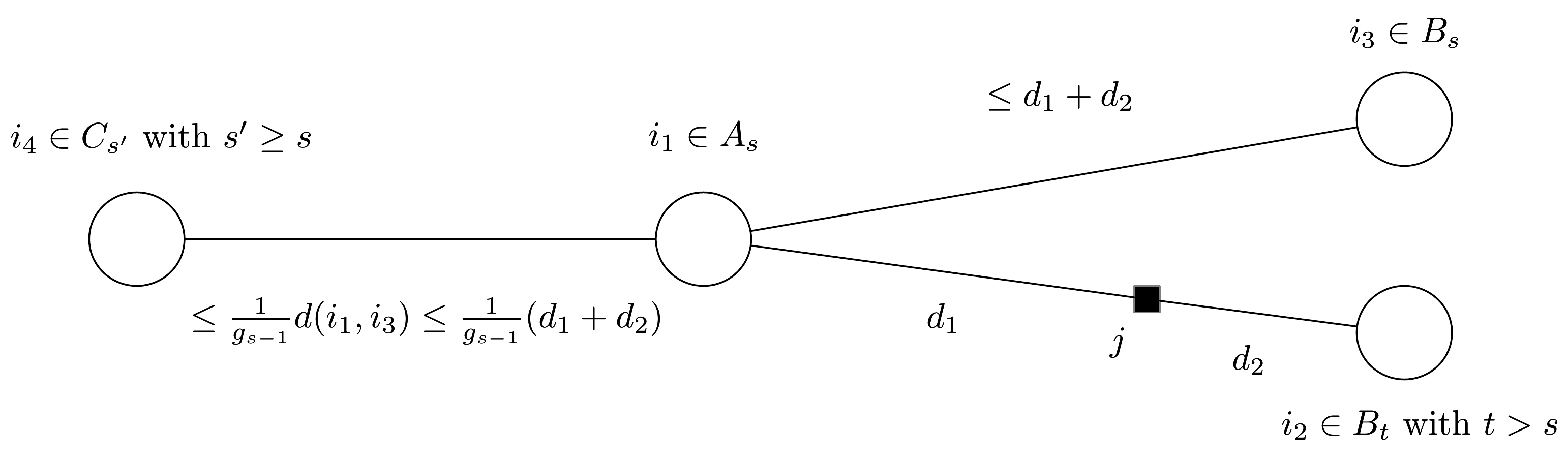}
\caption{Case 3.}
\end{figure}

    Therefore, by connecting to the first available facility in the order of precedence $(i_2,i_1,i_3,i_4)$:
    \begin{align*}
        \E[c(j)] &\le \Pr[i_2]d_2 + \Pr[\bar{i_2}i_1]d_1 + \Pr[\bar{i_2}\bar{i_1}i_3]d(j,i_3) + \Pr[\bar{i_2}\bar{i_1}\bar{i_3}]d(j,i_4)\\
         &= \Pr[i_2]d_2 + \Pr[\bar{i_2}]\Pr[i_1]d_1 + \Pr[\bar{i_2}]\Pr[\bar{i_1}]\Pr[i_3]d(j,i_3) + \Pr[\bar{i_2}]\Pr[\bar{i_1}]\Pr[\bar{i_3}]d(j,i_4)\\
         &\le \Pr[i_2]d_2 + \Pr[\bar{i_2}]\Pr[i_1]d_1 + \Pr[\bar{i_2}]\Pr[\bar{i_1}](d_1 + (d_1+d_2) + \Pr[\bar{i_3}]\left(\frac{1}{g_{s-1}}-1\right)(d_1+d_2) )\\
         &= \Pr[i_2]d_2 + \Pr[\bar{i_2}]d_1 + \Pr[\bar{i_2}]\Pr[\bar{i_1}](1 + \left(\frac{1}{g_{s-1}}-1\right)\Pr[\bar{i_3}])(d_1+d_2) = \C_{1,\frac{1}{g}}(j).
    \end{align*}

    In case (2), if $i_2 \ne i_3$, we get the same bound on the cost as in item (3). However, if $i_2 = i_3$ (which is possible since $i_3 \in \cup_{t=1}^s B_t$), we have only one guaranteed backup, i.e., $i_4$. In this case, by connecting to the first available facility in the order of precedence $(i_2,i_1,i_4)$:
    \begin{align*}
        \E[c(j)] &\le \Pr[i_2]d_2 + \Pr[\bar{i_2}i_1]d_1 + \Pr[\bar{i_2}\bar{i_1}]d(j,i_4)\\
         &= \Pr[i_2]d_2 + \Pr[\bar{i_2}]\Pr[i_1]d_1 + \Pr[\bar{i_2}]\Pr[\bar{i_1}]d(j,i_4)\\
         &\le \Pr[i_2]d_2 + \Pr[\bar{i_2}]\Pr[i_1]d_1 + \Pr[\bar{i_2}]\Pr[\bar{i_1}]\left(d_1 + \frac{1}{g_{s-1}}(d_1+d_2) \right)\\
         &= \Pr[i_2]d_2 + \Pr[\bar{i_2}]d_1 + \Pr[\bar{i_2}]\Pr[\bar{i_1}]\frac{1}{g_{s-1}}(d_1+d_2) = \C_{\frac{1}{g}}(j).
    \end{align*}
    Therefore, since $\C_{1,\frac{1}{g}}(j) \le \C_{\frac{1}{g}}(j)$, $\E[c(j)] \le \C_{\frac{1}{g}}(j)$.
    
\begin{figure}[h]
\centering 
\includegraphics[width=0.75\textwidth]{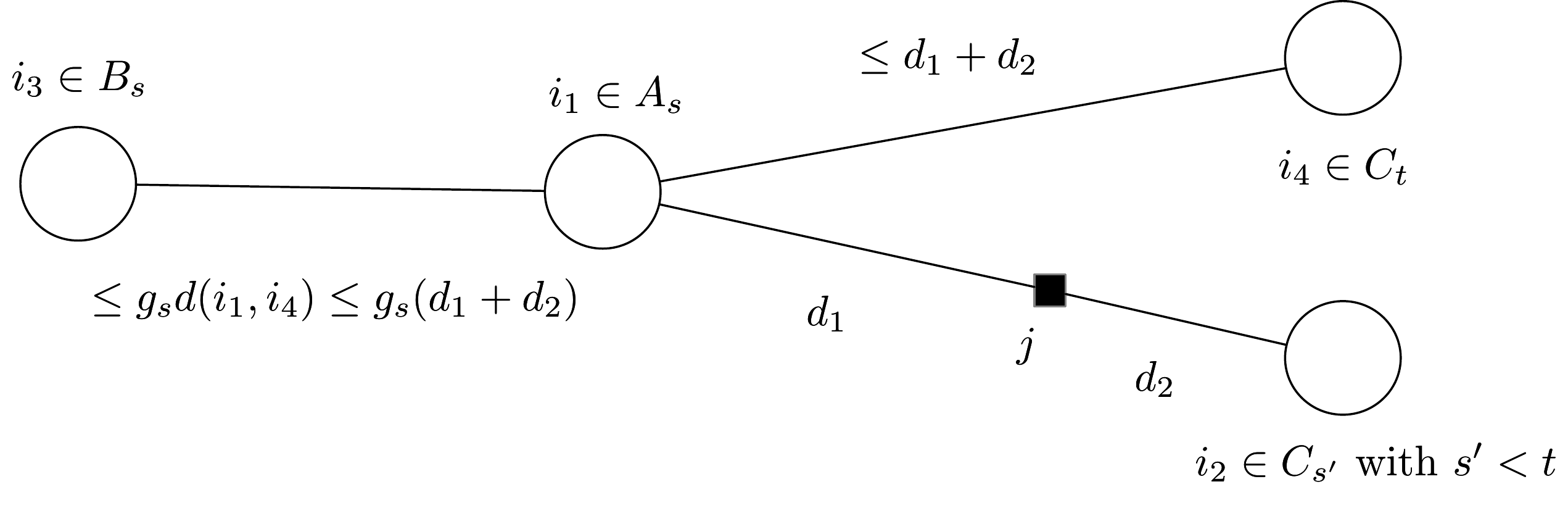}
\caption{Case 4.}
\end{figure}

    Finally, for case (4), since $i_1 \in A_s$, $g(i_1) \le g_s \implies d(i_1,i_3)\le g_s d(i_1,i_4)$. And, since $i_2 \in C_t$, $d(i_1,i_4) \le d(i_1,i_2) \le d_1+d_2$. Hence, by triangle inequality, $d(j,i_3) \le d(j,i_1)+d(i_1,i_3) \le d_1 + g_s(d_1+d_2)$. Further, $d(j,i_4) \le d(j,i_1) + d(i_1,i_4) \le 2d_1+d_2$. Therefore, by connecting to the first available facility in the order of precedence $(i_2,i_1,i_3,i_4)$:
    \begin{align*}
        \E[c(j)] &\le \Pr[i_2]d_2 + \Pr[\bar{i_2}i_1]d_1 + \Pr[\bar{i_2}\bar{i_1}i_3]d(j,i_3) + \Pr[\bar{i_2}\bar{i_1}\bar{i_3}]d(j,i_4)\\
         &= \Pr[i_2]d_2 + \Pr[\bar{i_2}]\Pr[i_1]d_1 + \Pr[\bar{i_2}]\Pr[\bar{i_1}]\Pr[i_3]d(j,i_3) + \Pr[\bar{i_2}]\Pr[\bar{i_1}]\Pr[\bar{i_3}]d(j,i_4)\\
         &\le \Pr[i_2]d_2 + \Pr[\bar{i_2}]\Pr[i_1]d_1 + \Pr[\bar{i_2}]\Pr[\bar{i_1}](d_1 + g_s(d_1+d_2) + \Pr[\bar{i_3}](1-g_s)(d_1+d_2) )\\
         &= \Pr[i_2]d_2 + \Pr[\bar{i_2}]d_1 + \Pr[\bar{i_2}]\Pr[\bar{i_1}](g_s + \Pr[\bar{i_3}](1-g_s))(d_1+d_2) = \C_{g,1}(j).
    \end{align*}

Note we have implicitly assumed in several cases that certain facilities are distinct, which may not be the case. It could be that $i_2=i_3$ in case (1), or $i_2 = i_4$ in case (4), for example. However, it is straightforward to show for these cases that the cost will only be less than the given bound.
    
\end{proof}

\subsection{The Factor-Revealing NLP}\label{sec:nlp}
We now construct the non-linear program (NLP) that bounds the bi-point rounding factor of our algorithm. We follow a similar approach as in Byrka et al.~\cite{Byrka-et-al}. The objective of the NLP is to maximize the bi-point rounding factor of a bi-point solution (i.e., the ratio between the total connection cost achieved by the solution returned by our algorithm and the bi-point solution cost), over the space of all bi-point solutions.

First, we define the following parameters which we will use to bound the total cost of each algorithm.
Consider the partition of the clients according to the memberships of their corresponding closest facilities in $F_1$ and $F_2$ as follows:
\begin{align*}
    \mathcal{J}_B^{x,y} &:= \{j \in  \mathcal{C} | i_1 \in A_x \text{ and } i_2\in B_y\},\\
    \mathcal{J}_C^{x,y} &:= \{j \in  \mathcal{C} | i_1 \in A_x \text{ and } i_2\in C_y\}.
\end{align*}
For each client class defined above and for $Z \in \{B,C\}$, we define: 
\begin{align*}
    & D_{Z, 1}^{x,y} = \sum_{j\in \mathcal{J}_Z^{x,y}} d_1(j),\\
    & D_{Z, 2}^{x,y} = \sum_{j\in \mathcal{J}_Z^{x,y}} d_2(j).
\end{align*}
Observe that $D_1 = \sum_{x\in[m]}\sum_{y\in [m]} D_{B,1}^{x,y}+D_{C,1}^{x,y}$ and $D_2 = \sum_{x\in[m]}\sum_{y\in [m]} D_{B,2}^{x,y}+D_{C,2}^{x,y}$.
We now define the corresponding aggregate versions of the cost functions from \Cref{sec:cost}.
\begin{align*}
    & 1.\,\C_1(\mathcal{J}_B^{x,y}) =\: p_{B_y}D_{B,2}^{x,y} + (1-p_{B_y})D_{B,1}^{x,y} + (1-p_{B_y})(1-p_{A_x})(D_{B,1}^{x,y}+D_{B,2}^{x,y}),\\
    & 2.\,\C_\frac{1}{g}(\mathcal{B}_Z^{x,y}) =\: p_{B_y}D_{B,2}^{x,y} + (1-p_{B_y})D_{B,1}^{x,y} + \frac{1}{g_{x-1}}(1-p_{B_y})(1-p_{A_x})(D_{B,1}^{x,y}+D_{B,2}^{x,y}),\\
    & 3.\,\C_{1,\frac{1}{g}}(\mathcal{J}_B^{x,y}) =\: p_{B_y}D_{B,2}^{x,y} + (1-p_{B_y})D_{B,1}^{x,y}\\
    & \hspace{2cm}+ (1-p_{B_y})(1-p_{A_x})(1+\Big(\frac{1}{g_{x-1}}-1\Big)\big(1-\min_{s \in [x]} p_{B_s}\big))(D_{B,1}^{x,y}+D_{B,2}^{x,y}),\\
    & 4.\,\C_{g,1}(\mathcal{J}_C^{x,y}) =\: p_{C_y}D_{C,2}^{x,y} + (1-p_{C_y})D_{C,1}^{x,y}\\
     & \hspace{2cm}+ (1-p_{C_y})(1-p_{A_x})(g_{x}+(1-g_{x})\big(1-\min_{s \in [x]} p_{B_s}\big))(D_{C,1}^{x,y}+D_{C,2}^{x,y}).
\end{align*}
Finally, we define
\begin{align}\label{eq:cost_alg}
    cost(\mathcal{A}) :=& \sum_{y\in[m]}\C_1(\mathcal{J}_B^{1,y}) + \sum_{\substack{x\in[2,m],\\y\in[x]}}\C_{\frac{1}{g}}(\mathcal{J}_B^{x,y}) + \sum_{\substack{x\in[2,m-1],\\ y\in[x+1,m]}} \C_{1,\frac{1}{g}}(\mathcal{J}_B^{x,y}) + \sum_{\substack{x\in[m],\\y\in[m]}} \C_{g,1}(\mathcal{J}_C^{x,y}).
\end{align}

\begin{lemma}
    For each algorithm $\mathcal{A} \in \mathcal{ALG}_m$, the total expected cost is bounded above by $cost(\mathcal{A})$.
\end{lemma}
\begin{proof}
    Given an algorithm $\mathcal{A}(p_{A_1},\cdots,p_{A_{m}},p_{B_1},\cdots,p_{B_{m}},p_{C_1},\cdots,p_{C_{m}})$ in $\mathcal{ALG}_m$, each facility $i\in W$ is opened with probability $\Pr[i]=\frac{\lceil p_W|W|\rceil}{|W|}=p_W$ (by \Cref{claim:num_fac}). Thus for $j\in \mathcal{J}_B^{x,y}$, $\Pr[i_1]=p_{A_x}$ and $\Pr[i_2]=p_{B_y}$ (and similarly for $\mathcal{J}_C^{x,y}$). By construction $i_3:=\sigma_B(i_1) \in \cup_{s=1}^t B_s$, so we may say $\Pr[i_3]\ge \min_{s \in [t]} p_{B_s}$. The rest of the proof follows straightforwardly by summing the cost bounds obtained in \Cref{lem:client_cost} over each corresponding client class and by linearity of expectation.
\end{proof}

As shown in the subsequent section, all algorithm parameters ($p_W$ for $W \in P$) may be expressed in terms of $b$ and $\gamma_{A_1},\ldots, \gamma_{A_m}$ (and $\gamma_{C_1},\ldots,\gamma_{C_m}$, which can be rewritten in terms of $\gamma_{A_1},\ldots, \gamma_{A_m}$ by definition). Thus, $cost(\mathcal{A})$ is a function of $b,\gamma_{A_1},\ldots, \gamma_{A_m}$, $D_{Z,1}^{x,y}$ and $D_{Z,2}^{x,y}$ (as defined above), and $\g$. Recall that $\g$ is a vector of constants chosen by us at the start of the algorithm. Let the remaining parameters (including $D_1$ and $D_2$) be variables chosen by an adversary. Then we may bound the bi-point rounding factor by the following program.

\begin{align}\label{eq:NLP}
    \text{maximize} &\:\:\: X\\\label{nlp:cons1}
    \text{s.t.} &\:\:\: X \le cost(\mathcal{A}) & \forall \mathcal{A}\in \mathcal{ALG}_m\\\label{nlp:cons2}
     &\:\:\: X \le (1-b)D_1+b(3-2b)D_2 & cost(\mathcal{SR}) \text{, from Theorem~\ref{thm:sr}}\\\label{nlp:cons3}
     &\:\:\: (1-b)D_1 + bD_2 = 1 \\\label{nlp:cons4}
     &\:\:\: D_1 = \sum_{x,y\in[m]} D_{B,1}^{x,y}+D_{C,1}^{x,y}\\\label{nlp:cons5}
     &\:\:\: D_2 = \sum_{x,y\in[m]} D_{B,2}^{x,y}+D_{C,2}^{x,y}\\\label{nlp:cons6}
     &\:\:\: D_2 \le D_1\\\label{nlp:cons7}
     &\:\:\: 0 \le D_{Z,1}^{x,y}, 0 \le D_{Z,2}^{x,y} & \forall x \in [m], y\in [m], Z\in \{B,C\}\\
     &\:\:\: 0\le b\le 1 \\
     &\:\:\: 0 \le \gamma_{A_t} & \forall t \in [m]
\end{align}

\begin{lemma}\label{lem:algo_bound}
Given a bi-point solution, the expected cost of the best solution returned by algorithms $\mathcal{SR}$ and $\mathcal{ALG}_m$ is at most $X^*\cdot(1+\epsilon)$ times the cost of the bi-point solution, where $X^*$ is the solution to the above NLP~(\ref{eq:NLP}).
\end{lemma}
\begin{proof}
Let $X$ denote the cost of the cheapest solution output by $\mathcal{ALG}_m \cup \mathcal{LS}$. Constraints (\ref{nlp:cons1}) and (\ref{nlp:cons2}) enforce that $X$ is indeed the cheapest. By uniformly scaling all distances, we may first normalize the bi-point solution cost to $1$ (enforcing constraint (\ref{nlp:cons3})), so that $X$ is also the bi-point rounding factor. Constraints (\ref{nlp:cons4}) and (\ref{nlp:cons5}) must hold since the corresponding classes partition all clients. Constraint (\ref{nlp:cons6}) may be assumed, as otherwise we may trivially take $F_1$ as a solution with bi-point rounding factor less than one. Finally, we lose an additional multiplicative factor of $(1+\epsilon)$ if the best solution is returned by $\mathcal{SR}$, thereby giving us the stated bound.
\end{proof}

\subsubsection{Explicit description of algorithm parameters}\label{sec:desc_alg_params} 
In this section we will explicitly describe $\ALG_m$ for several small values of $m$. This will allow us to fully expand constraints (\ref{nlp:cons1}) in the NLP, after which we proceed to calculate a rigorous upper-bound.

As a concrete example, consider $m=1$. Here the facilities are partitioned into $P=\{A_1,B_1,C_1\}$, and each algorithm in $\ALG_1$ is of the form $\A(p_{A_1},p_{B_1},p_{C_1})$. Now, let us enumerate $\ALG_1$, as discussed in \Cref{sec:algo_def}. First, we restate the validity constraint \ref{eq:mass_eqn1} in terms of the NLP variables, by dividing both sides by $|C|$ and using the fact that $|A_t|=|B_t|$:
\begin{align}\label{eq:mass_normalized}
    \sum_{t=1}^{m} (p_{A_t}+p_{B_t})\gamma_{A_t} + \sum_{t=1}^{m} p_{C_t}\gamma_{C_t} = \sum_{t=1}^{m} \gamma_{A_t} + b.
\end{align}
Since $m=1$, we have $\gamma_{C_1}=1$, and the above simplifies to $(p_{A_1}+p_{B_1})\gamma_{A_1}+p_{C_1}=\gamma_{A_1} + b$. Now consider all possible ways of assigning two of the parameters ($p_{A_1}$, $p_{B_1}$, and $p_{C_1}$) to $0$ or $1$ and setting the third such that (\ref{eq:mass_normalized}) is satisfied. This results in $12$ potential algorithms. After filtering out algorithms which fail to set at least one of $p_{A_1}$ or $p_{B_1}$ to $1$ (\Cref{def:valid_algos} property $4$), as well as those where the fractional argument is never between $0$ and $1$ (\Cref{def:valid_algos} property $1$), there are $5$ (conditionally) valid algorithms remaining, shown in \Cref{tab:alg_1}. Notice $\mathcal{A}_3$ is only valid when $\gamma_{A_1}\le b$, and $\A_4$ and $\A_5$ are valid only when $\gamma_{A_1}\ge b$. We may account for this either by formulating and solving a separate NLP for each case, or by cleverly combining the algorithms (as is later shown). Finally, we may explicitly express constraints (\ref{nlp:cons1}) by calculating, for example:
\begin{align*}
cost(\A_1)&=\C_1(\mathcal{J}_B^{1,1})+\C_{g,1}(\mathcal{J}_C^{1,1})=D_{B,2}^{1,1} + b D_{C,2}^{1,1} +(1-b)(2D_{C,1}^{1,1}+D_{C,2}^{1,1})
\end{align*}
\begin{table}[htpb]
    \centering
    \setlength{\tabcolsep}{8pt}
    \renewcommand{\arraystretch}{1.5}
    \begin{tabular}{c|c|c|c}
    Algorithm       & $p_{A_1}$ & $p_{B_1}$ & $p_{C_1}$  \\[3pt]\hline
    $\mathcal{A}_1$ &     $0$     &     $1$    &     $b$   \\[3pt]\hline
    $\mathcal{A}_2$ &     $1$     &     $0$    &     $b$   \\[3pt]\hline
    $\mathcal{A}_3$ &     $1$     &     $1$    & $b-\gamma_{A_1}$ \\[3pt]\hline
    $\mathcal{A}_4$ & $\frac{b}{\gamma_{A_1}}$ &  $1$  &  $0$  \\[3pt]\hline
    $\mathcal{A}_5$ &     $1$     & $\frac{b}{\gamma_{A_1}}$ &     $0$   \\[3pt]\hline
    \end{tabular}
    \vspace{3pt}
    \caption{$\mathcal{ALG}_1$ }\label{tab:alg_1}
    \end{table}
For $m=1$, we omit a rigorous upper bound since the NLP does not even improve upon the $\frac{1+\sqrt{3}}{2}$ ratio attained by LS. This may be proved by checking that the following is a feasible solution: $b=\frac{3-\sqrt{3}}{2}$, $\gamma_{A_1}\to\infty$, $D_{B,1}^{1,1}=\frac{1}{\sqrt{3}}$, $D_{B,2}^{1,1}=0$, $D_{C,1}^{1,1}=D_{C,2}^{1,1}=D_2=\frac{3+\sqrt{3}}{6}$, and $X=D_1=\frac{1+\sqrt{3}}{2}$. (We also remark that as $\gamma_{A_1}$ approaches $\infty$, $\ALG_1$ reduces to $\{\A_1,\A_2\}$, which is identical to the pair of algorithms considered in JV.)

For larger $m$, we may somewhat simplify the analysis by only utilizing a partial set of algorithms $\ALG'_m\subseteq\ALG_m$ in our analysis. Note that substituting any such $\ALG'_m$ into the NLP produces a valid relaxation (since it is equivalent to removing the constraints corresponding to $\ALG_m\setminus \ALG'_m$). Thus, we need only prove validity of $\ALG'_m$ rather than completeness. 

\begin{table}[htpb]
\centering
\setlength{\tabcolsep}{8pt}
\renewcommand{\arraystretch}{1.5}
\begin{tabular}{c|c|c|c|c|c|c}
Algorithm       & $p_{A_1}$ & $p_{A_2}$ & $p_{B_1}$ & $p_{B_2}$ & $p_{C_1}$ & $p_{C_2}$ \\[3pt]\hline
$\mathcal{A}_1$ &     $1$     &     $1$     &     $0$     &     $0$     &   $\frac{b-\gamma_{C_2}}{1-\gamma_{C_2}}$    &      $\frac{b}{\gamma_{C_2}}$   \\[3pt]\hline
$\mathcal{A}_2$ &     $0$     &     $0$     &     $1$     &     $1$     &   $\frac{b-\gamma_{C_2}}{1-\gamma_{C_2}}$    &      $\frac{b}{\gamma_{C_2}}$   \\[3pt]\hline
$\mathcal{A}_3$ &     $1$     &     $1$     &     $0$     &     $0$     &   $\frac{b}{1-\gamma_{C_2}}$    &      $\frac{b+\gamma_{C_2}-1}{\gamma_{C_2}}$ \\[3pt]\hline
$\mathcal{A}_4$ &     $0$     &     $1$     &     $1$     &     $0$     &   $\frac{b}{1-\gamma_{C_2}}$    &      $\frac{b+\gamma_{C_2}-1}{\gamma_{C_2}}$ \\[3pt]\hline
$\mathcal{A}_5$ &     $0$     &     $0$     &     $1$     &     $1$     &   $\frac{b}{1-\gamma_{C_2}}$    &      $\frac{b+\gamma_{C_2}-1}{\gamma_{C_2}}$ \\[3pt]\hline
$\mathcal{A}_6$ &     $1$     &     $1$     &     $0$     &   $\frac{b}{\gamma_{A_2}}$    &  $\frac{b-\gamma_{A_2}}{1-\gamma_{C_2}}$    &    $0$  \\[3pt]\hline
$\mathcal{A}_7$ &     $0$     &     $1$     &     $1$     &   $\frac{b}{\gamma_{A_2}}$    &  $\frac{b-\gamma_{A_2}}{1-\gamma_{C_2}}$    &   $0$  \\[3pt]\hline
$\mathcal{A}_8$ &     $0$     &     $\frac{b}{\gamma_{A_2}}$     &     $1$     &     $1$     &   $\frac{b-\gamma_{A_2}-\gamma_{C_2}}{1-\gamma_{C_2}}$    &    $\frac{b-\gamma_{A_2}}{\gamma_{C_2}}$   \\[3pt]\hline
$\mathcal{A}_9$ &     $0$     &     $0$     &     $1$     &    $\frac{b+\gamma_{A_2}-1}{\gamma_{A_2}}$    &   $\frac{b+\gamma_{A_2}-\gamma_{C_2}}{1-\gamma_{C_2}}$    &   $\frac{b+\gamma_{A_2}}{\gamma_{C_2}}$    \\[3pt]\hline
$\mathcal{A}_{10}$ &     $0$     &    $\frac{b+\gamma_{A_2}-\gamma_{C_2}}{\gamma_{A_2}}$   &   $1$     &  $\frac{b-\gamma_{C_2}}{\gamma_{A_2}}$     &   $\frac{b-\gamma_{A_2}-\gamma_{C_2}}{1-\gamma_{C_2}}$  & $1$ \\[3pt]\hline
\end{tabular}
\vspace{3pt}
\caption{$\mathcal{ALG}'_2$ (all values truncated to $[0,1]$)}\label{tab:10_algs}
\end{table}

For $m=2$ and $m=3$, we define $\ALG'_2$ in \Cref{tab:10_algs} and $\ALG'_3$ in \Cref{tab:30_algs}. (See \Cref{sec:algogen} for discussion of heuristics used to pick these sets.) In order to avoid the conditional validity of most algorithms, and the resulting proliferation of NLPs, the algorithms are described in a certain form which is always valid. Note that all parameters listed in $\ALG'_2$ and $\ALG'_3$ are implied to be truncated to the unit interval as needed. As an example, let us fully state $\A_8$ from $\ALG'_2$:
\begin{align*}
    \mathcal{A}_8 = \Biggl(& p_{A_1}=0,p_{A_2}=\min\left\{1,\frac{b}{\gamma_{A_2}}\right\},p_{B_1}=1,p_{B_2}=1,\\
    & p_{C_1}=\max\left\{0,\frac{b-\gamma_{A_2}-\gamma_{C_2}}{1-\gamma_{C_2}}\right\},p_{C_2}=\min\left\{1,\max\left\{0,\frac{b-\gamma_{A_2}}{\gamma_{C_2}}\right\}\right\}\Biggr).
\end{align*}
To prove $\A_8$ is a valid algorithm, we may first restate it in a piecewise form, such that the parameters simplify in each case:
\[ \A_8=
\begin{cases}
    \A\left(p_{A_1}=0,p_{A_2}=\frac{b}{\gamma_{A_2}},p_{B_1}=1,p_{B_2}=1,p_{C_1}=0,p_{C_2}=0\right) & 0\le b\le \gamma_{A_2} \\
    \A\left(p_{A_1}=0,p_{A_2}=1,p_{B_1}=1,p_{B_2}=1,p_{C_1}=0,p_{C_2}=\frac{b-\gamma_{A_2}}{\gamma_{C_2}}\right) & \gamma_{A_2} \le b \le \gamma_{A_2}+\gamma_{C_2} \\
    \A\left(p_{A_1}=0,p_{A_2}=1,p_{B_1}=1,p_{B_2}=1,p_{C_1}=\frac{b-\gamma_{A_2}-\gamma_{C_2}}{1-\gamma_{C_2}},p_{C_2}=1\right) & \gamma_{A_2}+\gamma_{C_2} \le b \le 1
\end{cases} 
\]
In this form, it is straightforward to verify in each case that $\A_8$ has all of the properties in \Cref{def:valid_algos}, (using \Cref{eq:mass_normalized} for the second property), as well as having at most one fractional parameter. All other algorithms in $\ALG'_2$ and $\ALG'_3$ may be decomposed and proven valid in the same manner. 

Finally, since the resulting NLPs are non-convex and appear challenging to solve exactly, we employ the computer-assisted \emph{interval arithmetic} based approach of \cite{Zwick} to obtain a rigorous upper bound on the bi-point rounding factors. The approach is similar to~\cite{Byrka-et-al}, with some additional adjustments to handle the non-smooth parameter functions, and zeroes in the denominator. For further description, see \Cref{sec:interval_arithmetic}. Using these methods, we obtain the following results by \Cref{lem:algo_bound}.

\begin{theorem}
The expected cost of the best solution returned by $\SR$ and $\ALG_2$ with $g_1:=0.6586$ is at most $1.3103$ times the cost of the bi-point solution.
\end{theorem}
\begin{theorem}\label{thm:layer_3_resultb}
The expected cost of the best solution returned by $\SR$ and $\ALG_3$ with $g_1:=0.642$ and $g_2:=0.833$ is at most $1.3064$ times the cost of the bi-point solution.
\end{theorem}

Therefore, Theorem~\ref{thm:main_result} follows from Theorem~\ref{thm:layer_3_resultb}, \Cref{lem:star-round_fac} and \Cref{lem:main_fac}.

\section{Lower Bounds for our Framework}\label{sec:lower_bounds}
In this section, we will show that, given the NLP~(\ref{eq:NLP}), our hierarchy of partitioning schemes cannot achieve a bi-point rounding factor smaller than $1.2943$, proving Theorem~\ref{thm:lower_bound}. Note that this gives a rigorous lower bound only for our particular analysis, rather than the true performance of the algorithm.

Consider a bi-point solution such that $g(i) = \hat{g}$ for all $i \in F_1$ and some fixed value of $\hat{g}$. Observe that, regardless of how we choose the algorithm parameters $m$ and $g_1,g_2,\cdots,g_{m-1}$, all facilities in $F_1$ may
be assigned to a single set $A_t$. In the best case, we will have set $g_t=\hat{g}$ and $g_{t+1}=\hat{g}+\epsilon$ for small $\epsilon>0$, so that the cost functions in \Cref{sec:cost} are as tight as possible. WLOG, we may assume that $t=1$ and hence, the non-empty sets are $A_2, B_2, C_1,$ and $C_2$. 


Let $\mathcal{ALG}_{\text{uniform}}$ be the set of all valid algorithms over these sets. We provide a list of algorithms in \Cref{tab:lower_bound} (in the previously mentioned form), and claim these include all algorithms in $\mathcal{ALG}_{\text{uniform}}$.

To find a difficult instance, we essentially solved \ref{eq:NLP} heuristically, with an additional variable $g_1$ (i.e., treated as a variable instead of a parameter) in $[0, 1]$ and algorithms in $\mathcal{ALG}_{\text{uniform}} \cup \{\mathcal{SR}\}$. 


Now we implicitly consider a bi-point solution with parameters $\hat{g}=1, b=0.68, \gamma_{A_2}=0.7478,$ and $\gamma_{C_2}=0.3291$. Furthermore, distribute clients such that the cost parameters are $D_{B,1}^{2,2} = 0.722175$, $D_{C,1}^{2,1} = 0.647832$, $D_{C,1}^{2,2} = 0.317901$, $D_{B,2}^{2,2} = 0.289375$,  $D_{C,2}^{2,1} = 0.259589$, $D_{C,2}^{2,2} = 0.127384$. We claim it is straightforward to construct an instance with these parameters. We have argued above that the algorithms in \Cref{tab:lower_bound}, with parameters $g_1 = \hat{g}\approx g_2$ give the best possible approximation obtainable with $\mathcal{ALG}_m$. Calculating the cost, we see it obtains a bi-point rounding factor of $1.2943$.

\section{Integrality Gap for Bi-point Solutions}\label{sec:integrality_gap}

In this section, we demonstrate a family of bi-point solutions with integrality gap approaching $\sqrt{\phi}=\sqrt{\frac{1+\sqrt{5}}2}\approx1.272$ for large $k$, proving Theorem~\ref{thm:integrality_gap}.
\subsection{A Golden Bi-point Solution}
We first construct a bi-point solution $\B(k)$ and show that it is a valid bi-point solution with unit cost. To do this, we need the following constants for construction (given in several forms to facilitate later algebra). Let $\phi=\frac{1+\sqrt5}2\approx1.618$ be the golden ratio, and define $\omega:=\phi-\sqrt\phi$. Let us also define,
\begin{alignat*}{3}
\ell &:= \frac1\phi=\phi-1
&\approx0.618, \\
r_B&:=\omega\sqrt\phi
&\approx0.440,\\
r_C&:=(1-\omega)\sqrt\phi
&\approx0.832,\\
b&:=\frac{1-r_B}{r_C}
=\frac12(1+\omega)
\quad&\approx0.673, \\
a&:=1-b=\frac12(1-\omega)
&\approx0.327.
\end{alignat*}

Since we cannot create fractional facilities, the actual construction will choose $r_B$ and $r_C$ to be the closest rational approximations of the form $t/k$ where $t \in \mathbb{Z}_+$ to approximate the above irrational values. We then use these approximations to derive $b:=\frac{1-r_B}{r_C}$, $a:=1-b$, so that the hard validity requirements in \Cref{validity} are still satisfied exactly. Thus, the actual constants will deviate from above values by $O(1/k)$. For the ease of notation, let us ignore this additive error for now.

Now we construct $\B(k)$ as follows. Let $A$ and $C$ be sets of facilities of sizes $r_Bk$ and $r_Ck$, respectively. For each pair $(i_1, i_2) \in A \times C$, set the distance $d(i_1, i_2) = 2$, and place a client $j$ with $d(j, i_1)=2 - \ell$ and $d(j, i_2)=\ell$. Let $\J_A$ denote this set of clients and set the demand for each client to $u_j := 1/|\J_A|$.

Additionally, for each facility $i_1\in A$, add facility $\beta(i_1)$ at distance $2\ell$ from $i_1$ and client $j$ colocated at $\beta(i_1)$ (i.e.,  $d(j,i_1)=2\ell$ and $d(j,\beta(i_1))=0$). Denote the added clients and facilities by $\J_B$ and $B$, respectively. Set the demand for each client to $u_j=a/|J_B|$.

Finally, assign all other distances according to the resulting graph metric.

\begin{figure}[htpb]
\centering 
\includegraphics[width=0.7\textwidth]{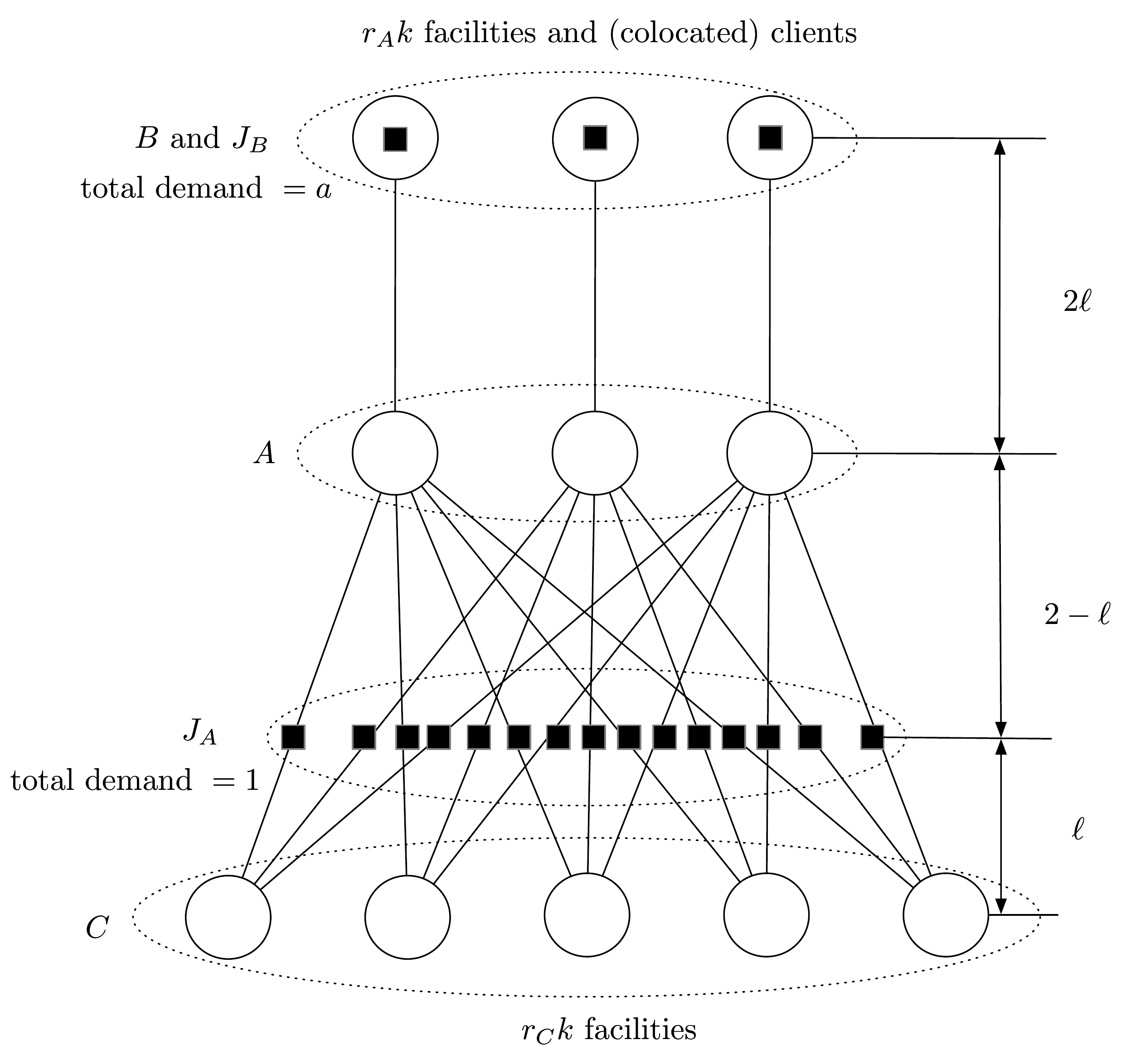}
\caption{The bi-point solution.}
\end{figure}

\begin{lemma}\label{validity}
The constructed bi-point solution $a\F_1+b\F_2$ where $\F_1 := A$ and $\F_2:=B\cup C$ is a valid bi-point solution.
\end{lemma}
\begin{proof}
By definition of our parameters, we have that
\begin{itemize}
    \item $a+b = 1$,
    \item $|\F_1| = |A| = r_Bk = \omega\sqrt\phi\cdot k<k$,
    \item $|\F_2| =|C|+|B|=(r_C+r_B)k=\sqrt\phi\cdot k>k$,
    \item $a|\F_1|+b|\F_2| = (1-b)r_Bk+bk(r_C+r_B)
=k(r_B+br_C)
=k\left(r_B+\frac{1-r_B}{r_C}r_C \right)=k.$
\end{itemize}
\end{proof}

\begin{lemma}
The constructed bi-point solution has unit cost.
\end{lemma}
\begin{proof}
We have that
\begin{align*}
    \text{Cost of } a\F_1+b\F_2 &=\sum_{j\in J_A}u_j\big(ad(j,A)+bd(j,\F_2)\big)
    +\sum_{j\in J_B}u_j\big(ad(j,A)+bd(j,\F_2)\big)\\
    &=\frac{1}{|J_A|}\sum_{j\in J_A}\big(a(2-\ell)+b\ell\big)
    +\frac{a}{|J_B|}\sum_{j\in J_B}\big(a\cdot2\ell+b\cdot0\big)\\
    &=a(2-\ell)+b\ell+2a^2\ell\\
    &=a(2-\ell)+(1-a)\ell+2a(1-b)\ell\\
    &=2a+\ell(1-2ab)\\
    &=(1-\omega)+\frac{1}{\phi}\Big(1-2\cdot\frac12(1-\omega)\cdot\frac12(1+\omega)\Big)\\
    &=1-\omega+\frac{1}{\phi}\Big(\frac12+\frac12\omega^2\Big)\\
    &=1+\frac{1}{2\phi}(1-2\phi\omega+\omega^2)\\
    &=1+\frac{1}{2\phi}(1+\phi-\phi^2)=1.
\end{align*}

\end{proof}

\textbf{Remarks.} The golden ratio appears visually in the instance as the ratio between $d(A,C)$ and $d(A,B)$. Thus, the instance is aesthetically pleasing in an objective sense. Also interestingly, the instance has the facility and cost ratios
\[
    \frac{|\F_1|}{|\F_2|} = \frac{r_Bk}{r_Bk+r_Ck} = \frac{\omega\sqrt{\phi}}{\sqrt{\phi}}=\omega,
\] \[
    \frac{D_2}{D_1} = \frac{\ell}{2a\ell+2-\ell} = \frac{1}{2a+2/\ell-1} = \frac{1}{-\omega+2\phi} = \omega,
\]
as $-\omega^2+2\omega\phi-1 = \phi^2-\phi-1 = 0$ by definition of $\phi$.


\subsection{Integrality Gap Of The Golden Bi-point Solution}
We proceed to prove that any solution $\S$ to $\B(k)$ must cost at least $\sqrt\phi-O(1/k)$. To this end, we will consider $\S$ in terms of the following parameters. Let $P:= \{(i,\beta(i)) : i \in A\}$ be the set of pairings between $A$ and $B$. We define
\begin{align*}
x_{00}:=&|\{(i_1,i_2)\in P:i_1\not\in\S\land i_2\not\in\S\}|/|A|, \\
x_{01}:=&|\{(i_1,i_2)\in P:i_1\not\in\S\land i_2\in\S\}|/|A|, \\
x_{10}:=&|\{(i_1,i_2)\in P:i_1\in\S\land i_2\not\in\S\}|/|A|, \\
x_{11}:=&|\{(i_1,i_2)\in P:i_1\in\S\land i_2\in\S\}|/|A|, \\
x_A:=&|A\cap\S|/|A|=x_{10}+x_{11}, \\
x_B:=&|B\cap\S|/|B|=x_{01}+x_{11}, \\
x_C:=&|C\cap\S|/|C|.
\end{align*}

\begin{lemma} These parameters obey the following constraints:
\begin{itemize}
    \item All parameters above lie in $[0,1]$,
    \item $x_{00} \ge1-x_B-x_A$,
    \item $x_Ar_B+ x_Br_B+x_Cr_C = 1$.
\end{itemize}
\end{lemma}

\begin{proof}

The first item follows by definition. Next, observe that $$x_{00}+x_{01}+x_{10}+x_{11} =|P|/|A|=1.$$ 

Adding $x_{11}$ to both sides and using the fact that $x_{01}+x_{11} = x_B$ and $x_{10}+x_{11} = x_A$, we get
$$x_{00}+x_B+x_A=1+x_{11}, $$
which implies that $x_{00} \geq 1-x_B-x_A$.

Finally, a feasible solution must have $|\S| \le k$. For purposes of determining a minimum cost solution, we may assume equality, as adding facilities to any smaller solution cannot increase the cost. Thus, we have

$$ |\S| = x_A|A|+ x_B|B|+x_C|C|
= k.$$
Dividing both sides by $k$ gives $x_Ar_B+ x_Br_B+x_Cr_C = 1.$
\end{proof}


\begin{lemma}
For any feasible solution $\S$ to $\B(k)$, we have that
\begin{align}
    c(\S) \geq  x+2(1-x_{C})(1-\ell x_A)
+2a\ell(1-x_B) + 2a\max\{1-x_B-x_A,0\}.
\label{cost}
\end{align}
\end{lemma}

\begin{proof}

Given $j\in J_A$, let $i_1\in A$, $i_2\in C$ be the facilities directly connected to $j$ during construction. So $i_2$ is closest at distance $\ell$ and $i_1$ is the second closest at distance $2-\ell$ from $j$. The third nearest facility is distance $\geq 2+\ell$ from $j$. For a random client $j\in J_A$:
\begin{align*}
\E\big[c(j)\mid j\in J_A\big]
&\ge\Pr[i_2\in\S]\ell+\Pr[i_2\not\in\S\land i_1\in\S](2-\ell)+\Pr[i_2\not\in\S\land i_1\not\in\S](2+\ell)\\
&= x_{C}\ell+(1-x_{C})x_A(2-\ell)+(1-x_{C})(1-x_A)(2+\ell)\\
&= x_{C}\ell+(1-x_{C})(2+\ell-2\ell x_A)\\
&= \ell+2(1-x_{C})(1-\ell x_A).
\end{align*}

Given $j\in J_B$, let $i_1\in A$, $i_2\in B$ be the closest facilities. Note that $i_2$ is closest at distance $0$ and $i_1$ is the second closest at distance $2\ell$ from $j$. All other facilities are at distance $\geq 2+2\ell$ from $j$. For a random client $j\in J_B$:
\begin{align*}
\E\big[c(j)\mid j\in J_B\big]
&\ge\Pr[i_2\in\S]\cdot0+\Pr[i_2\not\in\S\land i_1\in\S]\cdot2\ell+\Pr[i_2\not\in\S\land i_1\not\in\S]\cdot(2+2\ell)\\
&=\Pr[i_2\in\S]\cdot0+\Pr[i_2\not\in\S]\cdot 2\ell+\Pr[i_2\not\in\S\land i_1\not\in\S]\cdot 2\\
&= 2\ell(1-x_B) + 2x_{00}\\
&\ge 2\ell(1-x_B) + 2\max\{1-x_B-x_A,0\}.
\end{align*}

Thus, the total cost of the solution is
\begin{align}
c(\S)&=\sum_{j\in J_A\cup J_B} u_j c(j) \nonumber\\
&=\sum_{j\in J_A} \frac{1}{|J_A|}c(j)
+\sum_{j\in J_B} \frac{a}{|J_B|}c(j)\nonumber\\
&=\E\big[c(j)\mid j\in J_A\big]+a\E\big[c(j)\mid j\in J_B\big]\nonumber\\
&\ge \ell+2(1-x_{C})(1-\ell x_A)
+2a\ell(1-x_B) + 2a\max\{1-x_B-x_A,0\}.
\end{align}

\end{proof}

Now the following system $\mathcal{P}$ describes a lower bound on the cost of any feasible solution $\S$:
\begin{align}
\text{ min}\quad&\ell+2(1-x_C)(1-\ell x_A)+2a\ell(1-x_B) + 2a\max\{1-x_B-x_A,0\}\\
\text{s.t.}\quad&x_Ar_B+ x_Br_B+x_Cr_C=1\label{lpk}\\
&0\le x_A,x_B,x_C\le 1
\end{align}
\begin{lemma}\label{lem:golden-min-P}
For $\v=(x_A,x_B,x_C)\in \mathcal{P}$, the objective function $f(\v)$ is minimized at at least one extremal point (one with at most one fractional component).
\end{lemma}
\begin{proof}
Let $\v_0$ be a minimizer of $f$ in P, and initialize variable $\v:=\v_0$. We will demonstrate that we can move $\v$ to an extremal point without increasing $f(\v)$, proving the lemma. Define vectors $\e_1:=(1,-1,0)$, $\e_2:=(0,r_C,-r_B)$, $\e_3:=(r_C,0,-r_B)$. Observe movement along these vectors preserves constraint (\ref{lpk}).

First, observe $f(\v+t\e_1)$ is linear in $t$, therefore at least one direction of $\pm\e_1$ is non-increasing. Move $\v$ maximally in that direction, until a new constraint in $\mathcal{P}$
becomes tight (or if already tight then we can skip this movement). Now either $x_B$ or $x_A$ is integral (0 or 1), so $\max\{1-x_B-x_A,0\}$ can be simplified to be linear in the remaining variables: either $1-x_A$, $1-x_B$, or $0$.

Now, if $x_A$ is integral, let us move $\v$ in direction $\pm\e_2$ until $x_C$ or $x_B$ becomes integral. Again, $f(\v+t\e_2)$ is linear in $t$, so at least one direction will be non-increasing.

If instead $x_B$ is integral, then we move $\v$ in direction $\pm\e_3$ until one of the remaining variables becomes integral. $f(\v+t\e_3)$ is linear in $t$ except for term $-2\ell r_Br_Ct^2$ which is concave in $t$, therefore at least one direction will be non-increasing.

In either case, we have produced $\v$ with $f(\v)\le f(\v_0)$, and at most 1 fractional component.
\end{proof}

Now we need only enumerate the extremal points of $\mathcal{P}$ and find the minimum value of $f$. $\mathcal{P}$ is the intersection of a plane with a 3-dimensional unit cube, which can have up to $6$ extreme points. For our chosen parameters $r_B, r_C$, the extreme points $(x_A,x_C,x_B)$ are
\begin{alignat*}{3}
\v_1&= \left(1, \frac{1-r_B}{r_C}, 0 \right)=(1, b, 0)&\approx(1,0.673,0), \\
\v_2&= \left(0, \frac{1-r_B}{r_C}, 1\right)=(0, b, 1)&\approx(0,0.673,1), \\
\v_3&= \left(1,\frac{1-2r_B}{r_C},1\right) \quad&\approx(1,0.144 ,1), \\
\v_4&= \left(0,1,\frac{1-r_C}{r_B}\right)=(0,1,2-\phi)\quad&\approx(0,1,0.382), \\
\v_5&= \left(\frac{1-r_C}{r_B},1, 0\right)= (2-\phi,1,0)&\approx(0.382,1,0).
\end{alignat*}

Substituting these points to $f$ gives
\begin{align*}
f(\v_1)=f(1, b, 0)&=\ell+2(1-b)(1-\ell)+2a\ell=\ell+2a(1-\ell)+2a\ell=\ell+2a= \sqrt{\phi},
\end{align*}
\begin{align*}
f(\v_2)=f(0, b, 1) &= \ell+2(1-b) = \ell + 2a = \sqrt{\phi},
\end{align*}
\begin{align*}
f(\v_4)=f(0,1,2-\phi) &= \ell+2a\ell(\phi-1)+2a(\phi-1)\\
&= \frac{1+2a\ell+2a}{\phi} \\
&= \frac{1+(1-\omega)(\phi-1)+(1-\omega)}{\phi}  \\
&= \frac{1+\phi-\phi^2 + \phi \sqrt{\phi}}{\phi}  \\
&= \sqrt{\phi},
\end{align*}
and numerically,
$$ f(\v_5)=f(2-\phi,1,0) = \frac{3}{\phi}+\frac{2}{\sqrt\phi}-2\approx1.426 > \sqrt{\phi}.$$
Also,
\begin{align*}
f(\v_3)=f(1, (1-2r_B)/2_C, 1) &= \ell +2\left(1-\frac{1-2r_B}{r_C}\right)(1-\ell) \\
&= \phi-1 + 2 \frac{(1-\omega)\phi - \sqrt{\phi} + 2\omega \phi}{(1-\omega)\phi}(2-\phi) \\
&= \phi-1 + 2 \frac{\omega(1+ \phi)}{(1-\omega)\phi}(2-\phi).
\end{align*}
We claim that the RHS is exactly $\sqrt{\phi}$ as
\begin{align*}
  \phi-1 + 2 \frac{\omega(1+\phi)}{(1-\omega)\phi}(2-\phi) &=  \sqrt{\phi} \\
  \Leftrightarrow (\phi-1)(1-\omega) + 2\omega(\phi-1) &= \sqrt{\phi}(1-\omega) \\
  \Leftrightarrow 2\omega(\phi-1)  &= (1-\omega)(\sqrt{\phi}-\phi+1) \\
  \Leftrightarrow 2\omega\phi -  2\omega  &= 1-2\omega+\omega^2 \\
  \Leftrightarrow 2\omega\phi   &= 1 + \omega^2 \\
  \Leftrightarrow 2\phi^2-2\phi\sqrt{\phi}   &= 1 + \phi^2 - 2\phi\sqrt{\phi} + \phi \\
  \Leftrightarrow \phi^2-\phi-1  &= 0, 
\end{align*}
which is true by definition of $\phi$.

In summary, $f(\v_5)>f(\v_1)=f(\v_2)=f(\v_3)=f(\v_4)=\sqrt\phi\approx1.272$. Therefore, the optimization program $\mathcal{P}$ is minimized at $\sqrt\phi$. Since the fractional cost of $\B(k)$ is 1, this implies the original $\I(k)$ has an integrality gap of $\sqrt\phi$.

\subsubsection{Final Gap}
As discussed earlier, our instance cannot represent the irrational constants $r_B$ and $r_C$ exactly, and so we instead set those constants to rational approximations with error $O(1/k)$. We claim this should only affect the solution of $\mathcal{P}$  by $O(1/k)$, which is less than any $\eps>0$ for sufficiently large $k$.

Furthermore, suppose we allow solutions to open $k+o(k)$ facilities. Then we must only slightly relax constraint (\ref{lpk}) by adding an $o(k)/k=o(1)$ term to the RHS. Again, the impact on the solution of $\mathcal{P}$ is bounded by any $\eps>0$  for sufficiently large $k$. 

\section{Discussion}\label{sec:discussion}
Our algorithm hierarchy arises from an attempt to increasingly tighten the $g(i)$-based bounds. Another way to utilize exact bounds could be with an appropriate randomized process, as was proposed at the end of \cite{Byrka-et-al}. Indeed our hierarchies may be viewed as increasingly-precise discretizations of such a process. In fact, if we fix all $\Omega$-variables in our NLP, and examine the dual of the resulting LP, we see that the dual solution represents an explicit probability distribution over our algorithms. This implies that for each $m$, there is a probability distribution (as a function of $\Omega$) over $\mathcal{ALG}_m$ which achieves the same approximation ratio as taking the best of all solutions. It would be interesting if there were a compact randomized algorithm that can provide the same result as $\lim_{m\to\infty}\mathcal{ALG}_m$ (as Jain and Vazirani essentially provided for $\mathcal{ALG}_1$).

We have shown that our algorithm hierarchy and analysis can prove a bi-point rounding factor between 1.2943 and 1.3064.
Our experiments loosely suggest that it may indeed achieve the lower bound as $m\to\infty$ in our hierarchy. However, even if true, we would still fail to match the known integrality gap.
Study of the integrality gap instance suggests that a matching approximation algorithm would need to leverage many more additional backup facilities per client, beyond the small number considered by current algorithms. The high-level idea is that if there are many nearby facilities, it should very likely to have a nearby backup bound---while if there are very few neighbors, then we can benefit from strong negative correlation with those few neighbors.

Star-rounding algorithms such as $\mathcal{SR}$ form a forest of $\F_1$-centric stars, and guarantee that each edge has an open endpoint. $\mathcal{ALG}_m$ does the same but with a forest of $\F_2$-centric stars. 
We claim that $\mathcal{SR}$ can actually be generalized to provide essentially the same guarantee for the forest of pseudo-trees (a tree plus one edge) which results from the union of both graphs. The fundamentals of the technique are similar to \cite{cohen2022improved}: trees can be broken up into constant-size components by removing a tiny fraction of nodes and/or edges.
This allows us to \emph{simultaneously} preserve both the \emph{root-or-leaf} guarantee of Li and Svensson, \emph{and} our ``$i_1$ or $\sigma_B(i_1)$" guarantee. This provides greatly reduced cost for the case when $\sigma_B(i_1)=i_2$, countering the costly $C_{1/g}(j)$ bound incurred by this case during $\mathcal{ALG}_m$. Although the overall improvement to our approximation factor appears quite small, we expect this may be a useful tool in matching the integrality gap of bi-point solutions. 

Ultimately, the new integrality gap shows that improving the bi-point factor alone cannot achieve better than $2\sqrt{\phi} \approx 2.544$. Improvement beyond this factor with bi-point solutions would require either improvement of the bi-point generation factor, or utilizing specifics of the bi-point generation algorithm for a more holistic analysis as in \cite{ahmadian2019bipoint,DBLP:conf/focs/CharikarG99}.

\paragraph*{Acknowledgments:} We thank the anonymous reviewers at SODA 2023 for their careful reading of our manuscript and valuable comments.

\bibliographystyle{plain}
\bibliography{references}
\appendix

\section{Computer-Assisted Techniques}\label{sec:comp_assist_techns} In this section, we describe the heuristics and computer-assisted techniques and implementation used to find the sets $\mathcal{ALG}'_2$ and $\mathcal{ALG}'_3$, and to obtain a rigorous upper bound on the bi-point rounding factor.

\subsection{Generating the Set of Algorithms}\label{sec:algogen}
Recall that $\mathcal{ALG}_m$ is the set of all valid algorithms with at most one fractional input parameter. As described in \Cref{sec:desc_alg_params} briefly, when solving the NLP, the set of parameters describing the bi-point solution are variable. Thus, $\mathcal{ALG}_m$ is a set of conditionally valid algorithms. This leads to a proliferation of NLPs, and we have to solve all of them in order to obtain the bi-point rounding factor. To avoid this, we construct algorithms that are piecewise valid (we refer to them as chains) (see the example in \Cref{sec:desc_alg_params}.) Chains can be considered as a collection of conditionally valid algorithms such that the validity range of these algorithms do not intersect and span the entire parameter space. It is straightforward to verify that chains satisfy all the properties in \Cref{def:valid_algos}, and hence they are valid.

Intuitively, each chain can be thought of as starting with some minimal feasible set of facilities in the partition $P$ (which is guaranteed to be smaller than $k$) and then opening additional facility sets in a pre-determined order, until we open $k$ in total. For the example stated in \Cref{sec:desc_alg_params} (i.e., $\A_8$ in \Cref{tab:10_algs}), the starting sets are $B_1$ and $B_2$, then the sets $A_2,C_2$ and $C_1$ are opened, in order. Thus, by enumerating all possible starting sets and orderings of the remaining facility sets, we can generate chains that guarantee to include all algorithms of $\mathcal{ALG}_m$ in at least one chain.

First, we have the following claim: every algorithm in $\mathcal{ALG}_m$ has at least $m$ input parameters that are equal to $1$. The proof follows from property (3) of \Cref{def:valid_algos}, wherein for every $t \in [m]$ at least one of $\{p_{A_t},p_{B_t},p_{C_t}\}$ is $1$.

Thus, by the above claim, we start by considering all possible $m$-sized subsets of $P$ which act as our starting sets and assign a value of $1$ to their corresponding parameters. Then, consider all possible orderings of the rest of the sets. Let $p_{W_1},p_{W_2},\cdots,p_{W_{2m}}$ be one such ordering. We now assign the fractional value obtained for $p_{W_1}$ by solving \Cref{eq:mass_normalized} and assuming all the remaining unassigned parameters to be $0$. This gives us the first piece of our chain. Next, assign the fractional value obtained for $p_{W_2}$ by solving \Cref{eq:mass_normalized} and assuming $p_{W_1}=1$ and all other unassigned parameters to be $0$. This gives the second piece of the chain. Continue this until the fractional value obtained is not a feasible quantity, i.e., we have run out of probability mass (whose total is $k$). In each step, make sure to check that the set of parameters form a valid algorithm. If it is not valid at any step, discard that ordering. Finally, replace every assigned fractional quantity $v$ with $\max\{0,\min\{1,v\}\}$. Repeat this process for all orderings, over all $m$-sized subsets to obtain the set of chains. Since we consider all possible $m$-sized subsets and all possible orderings of the rest of the sets, every valid algorithm will appear in at least one chain.

The resulting set of chains is excessive in the sense that the produced chains will have lots of overlap; indeed many chains will be entirely redundant and may be removed. We employ the following simple greedy heuristic: pick chains that cover the maximum number of uncovered valid algorithms iteratively. This produces a minimal set of chains such that every algorithm in $\mathcal{ALG}_m$ appears in at least one chain. For $m=2$ we get a set of $22$ chains, and for $m=3$ we get a set of $166$ chains, via this greedy heuristic.

Experimentally, we observe that not all constraints (corresponding to the chains) are tight. Thus, some of the chains can be dropped with little to no loss in the objective value of the NLP (this is valid since we get a relaxed NLP when we drop constraints (chains).) We employ the following heuristic we call the \emph{iterative addition} approach. We start with an initial set of chains ($\varnothing$ also works). In each iteration, solve the NLP with the current set, calculate the cost of the rest of the chains (not in the current set), and add the chain with the cheapest cost. We repeat until no new chain improves the cost. This heuristic helps reduce the number of chains by a significant amount. Namely, for $m=2$ we get $10$ chains (listed in \Cref{tab:10_algs}), and for $m=3$ we get $29$ chains (listed in \Cref{tab:30_algs}.)

\subsection{Bounding the NLP via Interval Arithmetic}\label{sec:interval_arithmetic}
Since the resulting NLPs are non-convex, it is difficult to solve them exactly. Thus, we employ the \emph{interval arithmetic} based approach of \cite{Zwick} to find a rigorous upper bound on the bi-point rounding factor. We follow a similar approach as in~\cite{Byrka-et-al}, with some additional adjustments to handle the non-smooth parameter functions, that arise due to chains, and division by $0$ in some cases. Let $\mathcal{ALG}'_m$ denote the set of chains obtained from the previous section. Also, let $\Omega$ denote the parameter space $\{0\le b \le 1\}\times \{0 \le \gamma_{A_1}\}\times \ldots \times \{0\le \gamma_{A_m}\}$.

First, consider the constraints corresponding to each algorithm $\mathcal{A}(p_{A_1}$, $\cdots$, $p_{A_m}$, $p_{B_1}$, $\cdots$, $p_{B_m}$, $p_{C_1}$, $\cdots$, $p_{C_m})$ in $\mathcal{ALG}'_m$, i.e., the constraint $X \le cost(\mathcal{A})$ (in \ref{eq:NLP}). Recall that $cost(\mathcal{A})$ is essentially a function of the input parameters $p_W$ for all $W \in P$~(\ref{eq:cost_alg}), and each $p_W$ is a function of non-linear variables~(\ref{eq:mass_normalized}). Let $\mathcal{I}$ be an interval in $\Omega$. Also, let 
\begin{align*}
    p_W^0 = \max\left\{0,\min_{\omega \in \mathcal{I}} p_W\right\},~~~~ p_W^1 = \min\left\{1,\max_{\omega \in \mathcal{I}} p_W\right\}.
\end{align*}
Then, let $cost'(\mathcal{A})$ be the value obtained by replacing the terms $p_W$ with $p_W^1$ and $1-p_W$ with $1-p_W^0$ in $cost(\mathcal{A})$. Observe that $\max_{\omega \in I} cost(\mathcal{A}) \le cost'(\mathcal{A})$. Therefore, we can relax the constraint $X \le cost(\mathcal{A})$ by substituting it with the constraint $X \le cost'(\mathcal{A})$.

Next, consider the normalization constraint $(1-b)D_1 + bD_2 = 1$. Let $b^0 = \min_{\omega \in \mathcal{I}} b$ and $b^1 = \max_{\omega \in \mathcal{I}} b$. This constraint can be simplified and relaxed to $D_2 + (D_1-D_2)(1-b) \le 1$. Then, given the interval $\mathcal{I}$, we substitute this constraint with the relaxation $D_2 + (D_1-D_2)(1-b^1) \le 1$.

Now, consider the constraint corresponding to the cost of $\mathcal{SR}$, i.e., $X \le (1-b)D_1+b(3-2b)D_2$. This constraint can be simplified as $X \le (1-b)D_1+bD_2 + 2b(1-b)D_2 \le 1+2b(1-b)D_2$ (since $(1-b)D_1+bD_2 \le 1$). Therefore, we can substitute this constraint with the relaxed constraint $X \le 1+2b^1(1-b^0)D_2$. Note that this is a tighter relaxation compared to $X \le (1-b^0)D_1+b^1(3-2b^0)D_2$; when $b\to 1$ and $D_2 \ll D_1$, the latter is unbounded.

Therefore, given an interval $\mathcal{I}$ in $\Omega$, the new and relaxed program obtained, by replacing the constraints in the NLP~(\ref{eq:NLP}) with their corresponding relaxations discussed, is an LP which can be solved efficiently with great precision. The value obtained by this LP will be an upper bound to our original NLP constrained to the interval $\mathcal{I}$. With sufficiently small intervals, we can obtain the value with desired precision. 

We run the interval arithmetic routine for the cases of $m=2$ with the chains listed in \Cref{tab:10_algs}, and $m=3$ with the chains listed in \Cref{tab:30_algs}. In our implementation, we start with $\Omega$ as our initial interval, i.e., $\mathcal{I}_0=\{b:[0,1],\gamma_{A_2}:[0,\infty),\cdots,\gamma_{A_m}:[0,\infty)\}$. Note that $\gamma_{A_1}$ doesn't appear in any of our chains, so we can drop it. In each iteration, if the LP corresponding to the interval achieves a value greater than our estimate (obtained by Mathematica's NLP solver), we split the interval into multiple sub-intervals ($8$ for $m=2$ and $16$ for $m=3$), typically dividing by half for each variable, and then solve the relaxed LP on each of the new sub-intervals. For an interval of the form $[0,\infty)$, we split it as $[0,N]$ and $[N,\infty)$, for some sufficiently large $N$ (in our implementation, $N=2$). The interval $[N,\infty)$ is not divided further.

Due to the nature of the chains, we may get terms involving $\frac{1}{0}$ and $\frac{0}{0}$ during the relaxation. We handle this by simply assigning $0$ when computing $p_W^0$ and $1$ when computing $p_W^1$, since they are the worst case lower and upper bounds for $p_W$. However, in the case when both the lower and upper bound calculation simplifies to the term $\frac{1}{0}$, we assign $p_W^0=1$ and $p_W^1=0$. This is valid since the denominator being $0$ essentially means that the corresponding set is empty, and the best lower and upper bounds would be $1$ and $0$, respectively. Alternatively, we could also just assign the cost variables corresponding to this set to be $0$ since the corresponding client sets are empty.

Using this approach, we obtain the value $1.3103$ for $m=2$ and $1.30634$ for $m=3$. This was implemented in \emph{Python~3} and we used IBM's CPLEX solver (version $20.1$) for solving the LPs. The $m=2$ case examined around $2.75$ million intervals and ran for around $36$ hours on a $12$-Core Intel Core i7 2.2 GHz machine. The $m=3$ case examined roughly $75$ million intervals and ran for around $10$ days on a $64$-Core 3rd Gen Intel Xeon 3.5 GHz machine. 

\section{\texorpdfstring{\boldmath$\mathcal{ALG}'_3$}{ALG'3}}\label{tab:30_algs}
All values are truncated to the interval $[0,1]$.
\begin{enumerate}
    \item $\mathcal{A}_1=(0$, $0$, $1$, $1$, $1$, $\frac{b-\gamma_{C_3}}{\gamma_{A_3}}$, $\frac{b-\gamma_{A_3}-\gamma_{C_3}}{1-\gamma_{C_2}-\gamma_{C_3}}$, $\frac{b+\gamma_{C_2}-1-\gamma_{A_3}}{\gamma_{C_2}}$, $\frac{b}{\gamma_{C_3}})$
    \item $\mathcal{A}_2=(0$, $\frac{b+\gamma_{A_2}+\gamma_{A_3}-\gamma_{C_2}-\gamma_{C_3}}{\gamma_{A_2}}$, $\frac{b+\gamma_{A_3}-1-\gamma_{A_2}}{\gamma_{A_3}}$, $1$, $\frac{b+\gamma_{A_3}-1}{\gamma_{A_2}}$, $0$, $\frac{b+\gamma_{A_3}-\gamma_{C_2}-\gamma_{C_3}}{1-\gamma_{C_2}-\gamma_{C_3}}$, $1$, $1)$
    \item $\mathcal{A}_3=(1$, $1$, $1$, $0$, $\frac{b-\gamma_{A_3}}{\gamma_{A_2}}$, $\frac{b}{\gamma_{A_3}}$, $\frac{b-\gamma_{A_2}-\gamma_{A_3}-\gamma_{C_2}-\gamma_{C_3}}{1-\gamma_{C_2}-\gamma_{C_3}}$, $\frac{b-\gamma_{A_2}-\gamma_{A_3}-\gamma_{C_3}}{\gamma_{C_2}}$, $\frac{b-\gamma_{A_2}-\gamma_{A_3}}{\gamma_{C_3}})$
    \item $\mathcal{A}_4=(0$, $1$, $1$, $1$, $\frac{b+\gamma_{C_2}-1-\gamma_{A_3}}{\gamma_{A_2}}$, $\frac{b+\gamma_{C_2}+\gamma_{C_3}-1}{\gamma_{A_3}}$, $\frac{b}{1-\gamma_{C_2}-\gamma_{C_3}}$, $0$, $\frac{b+\gamma_{C_2}+\gamma_{C_3}-1-\gamma_{A_3}}{\gamma_{C_3}})$
    \item $\mathcal{A}_5=(0$, $\frac{b-\gamma_{A_3}-\gamma_{C_3}}{\gamma_{A_2}}$, $\frac{b}{\gamma_{A_3}}$, $1$, $1$, $1$, $\frac{b-\gamma_{A_2}-\gamma_{A_3}-\gamma_{C_2}-\gamma_{C_3}}{1-\gamma_{C_2}-\gamma_{C_3}}$, $\frac{b-\gamma_{A_2}-\gamma_{A_3}-\gamma_{C_3}}{\gamma_{C_2}}$, $\frac{b-\gamma_{A_3}}{\gamma_{C_3}})$
    \item $\mathcal{A}_6=(0$, $\frac{b-\gamma_{C_3}}{\gamma_{A_2}}$, $\frac{b+\gamma_{A_3}-\gamma_{C_3}}{\gamma_{A_3}}$, $1$, $1$, $\frac{b-\gamma_{A_2}-\gamma_{C_3}}{\gamma_{A_3}}$, $\frac{b-\gamma_{A_2}-\gamma_{A_3}-\gamma_{C_3}}{1-\gamma_{C_2}-\gamma_{C_3}}$, $0$, $1)$
    \item $\mathcal{A}_7=(1$, $1$, $1$, $0$, $\frac{b+\gamma_{C_2}+\gamma_{C_3}-1}{\gamma_{A_2}}$, $0$, $\frac{b}{1-\gamma_{C_2}-\gamma_{C_3}}$, $\frac{b+\gamma_{C_2}+\gamma_{C_3}-1-\gamma_{A_2}}{\gamma_{C_2}}$, $\frac{b+\gamma_{C_3}-1-\gamma_{A_2}}{\gamma_{C_3}})$
    \item $\mathcal{A}_8=(1$, $1$, $1$, $0$, $\frac{b-\gamma_{C_2}}{\gamma_{A_2}}$, $\frac{b-\gamma_{A_2}-\gamma_{C_2}}{\gamma_{A_3}}$, $\frac{b-\gamma_{A_2}-\gamma_{A_3}-\gamma_{C_2}}{1-\gamma_{C_2}-\gamma_{C_3}}$, $\frac{b}{\gamma_{C_2}}$, $0)$
    \item $\mathcal{A}_9=(0$, $1$, $0$, $1$, $0$, $\frac{b+\gamma_{A_3}-1}{\gamma_{A_3}}$, $\frac{b+\gamma_{A_3}-\gamma_{C_3}}{1-\gamma_{C_2}-\gamma_{C_3}}$, $\frac{b+\gamma_{A_3}+\gamma_{C_2}-1}{\gamma_{C_2}}$, $1)$
    \item $\mathcal{A}_{10}=(0$, $1$, $1$, $1$, $0$, $\frac{b-\gamma_{C_2}}{\gamma_{A_3}}$, $\frac{b-\gamma_{A_3}-\gamma_{C_2}-\gamma_{C_3}}{1-\gamma_{C_2}-\gamma_{C_3}}$, $\frac{b}{\gamma_{C_2}}$, $\frac{b-\gamma_{A_3}-\gamma_{C_2}}{\gamma_{C_3}})$
    \item $\mathcal{A}_{11}=(0$, $\frac{b+\gamma_{C_2}+\gamma_{C_3}-1-\gamma_{A_3}}{\gamma_{A_2}}$, $1$, $1$, $1$, $\frac{b}{\gamma_{A_3}}$, $\frac{b-\gamma_{A_3}}{1-\gamma_{C_2}-\gamma_{C_3}}$, $0$, $0)$
    \item $\mathcal{A}_{12}=(1$, $1$, $1$, $0$, $\frac{b-\gamma_{A_3}-\gamma_{C_2}}{\gamma_{A_2}}$, $\frac{b-\gamma_{C_2}}{\gamma_{A_3}}$, $\frac{b-\gamma_{A_2}-\gamma_{A_3}-\gamma_{C_2}-\gamma_{C_3}}{1-\gamma_{C_2}-\gamma_{C_3}}$, $\frac{b}{\gamma_{C_2}}$, $\frac{b-\gamma_{A_2}-\gamma_{A_3}-\gamma_{C_2}}{\gamma_{C_3}})$
    \item $\mathcal{A}_{13}=(0$, $0$, $0$, $1$, $1$, $\frac{b+\gamma_{A_3}-\gamma_{C_3}}{\gamma_{A_3}}$, $\frac{b-\gamma_{C_3}}{1-\gamma_{C_2}-\gamma_{C_3}}$, $\frac{b+\gamma_{C_2}-1}{\gamma_{C_2}}$, $1)$
    \item $\mathcal{A}_{14}=(1$, $1$, $1$, $0$, $0$, $0$, $\frac{b-\gamma_{C_3}}{1-\gamma_{C_2}-\gamma_{C_3}}$, $\frac{b+\gamma_{C_2}-1}{\gamma_{C_2}}$, $\frac{b}{\gamma_{C_3}})$
    \item $\mathcal{A}_{15}=(0$, $0$, $\frac{b}{\gamma_{A_3}}$, $1$, $1$, $1$, $\frac{b-\gamma_{A_3}-\gamma_{C_2}-\gamma_{C_3}}{1-\gamma_{C_2}-\gamma_{C_3}}$, $\frac{b-\gamma_{A_3}}{\gamma_{C_2}}$, $\frac{b-\gamma_{A_3}-\gamma_{C_2}}{\gamma_{C_3}})$
    \item $\mathcal{A}_{16}=(1$, $1$, $1$, $0$, $0$, $0$, $\frac{b-\gamma_{C_2}-\gamma_{C_3}}{1-\gamma_{C_2}-\gamma_{C_3}}$, $\frac{b}{\gamma_{C_2}}$, $\frac{b-\gamma_{C_2}}{\gamma_{C_3}})$
    \item $\mathcal{A}_{17}=(0$, $0$, $0$, $1$, $\frac{b+\gamma_{A_2}+\gamma_{A_3}-\gamma_{C_2}-\gamma_{C_3}}{\gamma_{A_2}}$, $\frac{b+\gamma_{A_3}-\gamma_{C_2}-\gamma_{C_3}}{\gamma_{A_3}}$, $\frac{b-\gamma_{C_2}-\gamma_{C_3}}{1-\gamma_{C_2}-\gamma_{C_3}}$, $1$, $1)$
    \item $\mathcal{A}_{18}=(0$, $1$, $1$, $1$, $\frac{b}{\gamma_{A_2}}$, $\frac{b-\gamma_{A_2}}{\gamma_{A_3}}$, $\frac{b-\gamma_{A_2}-\gamma_{A_3}}{1-\gamma_{C_2}-\gamma_{C_3}}$, $0$, $0)$
    \item $\mathcal{A}_{19}=(0$, $1$, $1$, $1$, $\frac{b+\gamma_{C_2}-1}{\gamma_{A_2}}$, $0$, $\frac{b-\gamma_{C_3}}{1-\gamma_{C_2}-\gamma_{C_3}}$, $0$, $\frac{b}{\gamma_{C_3}})$
    \item $\mathcal{A}_{20}=(1$, $1$, $1$, $0$, $0$, $\frac{b+\gamma_{C_2}+\gamma_{C_3}-1}{\gamma_{A_3}}$, $\frac{b}{1-\gamma_{C_2}-\gamma_{C_3}}$, $\frac{b+\gamma_{C_2}+\gamma_{C_3}-1-\gamma_{A_3}}{\gamma_{C_2}}$, $0)$
    \item $\mathcal{A}_{21}=(0$, $0$, $0$, $1$, $\frac{b+\gamma_{A_2}+\gamma_{A_3}-1}{\gamma_{A_2}}$, $\frac{b+\gamma_{A_3}-1}{\gamma_{A_3}}$, $\frac{b+\gamma_{A_2}+\gamma_{A_3}-\gamma_{C_2}-\gamma_{C_3}}{1-\gamma_{C_2}-\gamma_{C_3}}$, $1$, $1)$
    \item $\mathcal{A}_{22}=(1$, $1$, $1$, $0$, $\frac{b-\gamma_{C_3}}{\gamma_{A_2}}$, $0$, $\frac{b-\gamma_{A_2}-\gamma_{C_3}}{1-\gamma_{C_2}-\gamma_{C_3}}$, $0$, $\frac{b}{\gamma_{C_3}})$
    \item $\mathcal{A}_{23}=(0$, $\frac{b+\gamma_{C_2}-1-\gamma_{A_3}}{\gamma_{A_2}}$, $\frac{b+\gamma_{C_2}-1}{\gamma_{A_3}}$, $1$, $1$, $1$, $\frac{b}{1-\gamma_{C_2}-\gamma_{C_3}}$, $0$, $\frac{b+\gamma_{C_2}+\gamma_{C_3}-1}{\gamma_{C_3}})$
    \item $\mathcal{A}_{24}=(0$, $1$, $1$, $1$, $\frac{b+\gamma_{C_3}-1}{\gamma_{A_2}}$, $0$, $\frac{b-\gamma_{C_2}}{1-\gamma_{C_2}-\gamma_{C_3}}$, $\frac{b}{\gamma_{C_2}}$, $\frac{b+\gamma_{C_3}-1-\gamma_{A_2}}{\gamma_{C_3}})$
    \item $\mathcal{A}_{25}=(0$, $\frac{b+\gamma_{A_2}-\gamma_{C_2}-\gamma_{C_3}}{\gamma_{A_2}}$, $\frac{b+\gamma_{A_2}+\gamma_{A_3}-\gamma_{C_2}-\gamma_{C_3}}{\gamma_{A_3}}$, $1$, $\frac{b-\gamma_{C_2}-\gamma_{C_3}}{\gamma_{A_2}}$, $0$, $\frac{b-\gamma_{A_2}-\gamma_{C_2}-\gamma_{C_3}}{1-\gamma_{C_2}-\gamma_{C_3}}$, $1$, $1)$
    \item $\mathcal{A}_{26}=(0$, $0$, $0$, $1$, $1$, $1$, $\frac{b-\gamma_{C_2}}{1-\gamma_{C_2}-\gamma_{C_3}}$, $\frac{b}{\gamma_{C_2}}$, $\frac{b+\gamma_{C_3}-1}{\gamma_{C_3}})$
    \item $\mathcal{A}_{27}=(1$, $1$, $1$, $0$, $0$, $\frac{b-\gamma_{C_2}}{\gamma_{A_3}}$, $\frac{b-\gamma_{A_3}-\gamma_{C_2}}{1-\gamma_{C_2}-\gamma_{C_3}}$, $\frac{b}{\gamma_{C_2}}$, $0)$
    \item $\mathcal{A}_{28}=(0$, $1$, $1$, $1$, $\frac{b-\gamma_{A_3}}{\gamma_{A_2}}$, $\frac{b}{\gamma_{A_3}}$, $\frac{b-\gamma_{A_2}-\gamma_{A_3}-\gamma_{C_3}}{1-\gamma_{C_2}-\gamma_{C_3}}$, $0$, $\frac{b-\gamma_{A_2}-\gamma_{A_3}}{\gamma_{C_3}})$
    \item $\mathcal{A}_{29}=(0$, $0$, $0$, $1$, $\frac{b+\gamma_{A_2}-1}{\gamma_{A_2}}$, $\frac{b+\gamma_{A_2}+\gamma_{A_3}-\gamma_{C_2}-\gamma_{C_3}}{\gamma_{A_3}}$, $\frac{b+\gamma_{A_2}-\gamma_{C_2}-\gamma_{C_3}}{1-\gamma_{C_2}-\gamma_{C_3}}$, $1$, $1)$
\end{enumerate}

\section{Table from \texorpdfstring{\Cref{sec:lower_bounds}}{Section 3}}
\begin{table}[htpb]
\centering
\setlength{\tabcolsep}{8pt}
\renewcommand{\arraystretch}{1.5}
\begin{tabular}{c|c|c|c|c|c|c}
Algorithm          & $p_{A_1}$ & $p_{A_2}$ & $p_{B_1}$ & $p_{B_2}$ &  $p_{C_1}$ & $p_{C_2}$ \\[3pt] \hline
$\mathcal{A}_1$    &   0         &    0       &      0     &     1      & $\frac{b}{1-\gamma_{C_2}}$ & $\frac{b+\gamma_{C_2}-1}{\gamma_{C_2}}$  \\[3pt] \hline
$\mathcal{A}_2$    &   0         &     0      &      0     &      $\frac{b+\gamma_{A_2}-1}{\gamma_{A_2}}$     &  $\frac{b+\gamma_{A_2}-\gamma_{C_2}}{1-\gamma_{C_2}}$ & 1 \\[3pt]  \hline
$\mathcal{A}_3$    &   0         &     0      &      0     &      $\frac{b + \gamma_{A_2} - \gamma{C_2}}{\gamma_{A_2}}$     &  $\frac{b - \gamma_{C_2} }{1 - \gamma_{C_2}}$ & 1 \\[3pt]  \hline
$\mathcal{A}_4$    &   0         &     1      &      0     &      0     &  $\frac{b}{1 - \gamma_{C_2}}$ & $\frac{b + \gamma_{C_2}-1}{\gamma_{C_2}}$ \\[3pt]  \hline
$\mathcal{A}_5$    &   0         &    1      &      0     &     $ \frac{ b + \gamma_{C_2}-1}{\gamma_{A_2}} $   &  $\frac{b}{1 - \gamma_{C_2}}$ & 0 \\[3pt]  \hline
$\mathcal{A}_6$    &   0         &     1     &      0     &      $\frac{b}{\gamma_{A_2}}$     &  $\frac{b - \gamma_{A_2}}{1 - \gamma_{C_2}}$ & 0 \\[3pt]  \hline
$\mathcal{A}_7$    &   0         &     1      &      0     &     $\frac{b}{\gamma_{A_2}}$    &  $\frac{b - \gamma_{A_2} - \gamma_{C_2}}{1 - \gamma_{C_2}}$ & $\frac{b - \gamma_{A_2}}{\gamma_{C_2}}$ \\[3pt]  \hline
$\mathcal{A}_8$    &   0         &     1      &      0     &      $\frac{b - \gamma_{C_2}}{\gamma_{A_2}}$     &  $\frac{b - \gamma_{A_2} - \gamma_{C_2}}{1 - \gamma_{C_2}}$ & $\frac{b}{\gamma_{C_2}}$  \\[3pt]  \hline
$\mathcal{A}_9$    &   0         &    $\frac{ b + \gamma_{A_2}-1}{\gamma_{A_2}}$      &      0     &      0     &  $\frac{b + \gamma_{A_2} - \gamma_{C_2}}{1 - \gamma_{C_2}}$ & 1 \\[3pt]  \hline
$\mathcal{A}_{10}$    &   0         &    $\frac{ b + \gamma_{C_2}-1}{\gamma_{A_2}}$      &      0     &      1     &  $\frac{b}{1 - \gamma_{C_2}}$ & 0 \\[3pt]  \hline
$\mathcal{A}_{11}$    &   0         &    $\frac{b}{\gamma_{A_2}}$      &      0     &      1     &  $\frac{b - \gamma_{A_2}}{1 - \gamma_{C_2}}$ & 0 \\[3pt]  \hline
$\mathcal{A}_{12}$    &   0         &     $\frac{b + \gamma_{A_2} - \gamma_{C_2}}{\gamma_{A_2}}$     &      0     &      0     &  $\frac{b - \gamma_{C_2}}{1 - \gamma_{C_2}}$ & 1 \\[3pt]  \hline
$\mathcal{A}_{13}$    &   0         &     $\frac{b - \gamma_{C_2}}{\gamma_{A_2}}$     &      0     &      1     &  $\frac{b - \gamma_{A_2} - \gamma_{C_2}}{1 - \gamma_{C_2}}$ & $\frac{b}{\gamma_{C_2}}$ \\[3pt]  \hline
$\mathcal{A}_{14}$    &   0         &     $\frac{b - \gamma_{C_2}}{\gamma_{A_2}}$      &      0     &      $\frac{b + \gamma_{A_2} - \gamma_{C_2}}{\gamma_{A_2}} $    &  $\frac{b - \gamma_{A_2} - \gamma_{C_2}}{1 - \gamma_{C_2}}$ & 1 \\[3pt]  \hline
\end{tabular}
\vspace{3pt}
\caption{Set of chains for all valid algorithms when $g_1 := \hat{g}, g_2 := \hat{g}+\epsilon$ for some small $\epsilon$, and $m = 2$.}\label{tab:lower_bound}
\end{table}
\end{document}